\documentclass{article}

\usepackage[preprint]{neurips_2025}

\usepackage[utf8]{inputenc}
\usepackage[T1]{fontenc}
\usepackage{url}
\usepackage{booktabs}
\usepackage{microtype}
\usepackage[dvipsnames]{xcolor}
\usepackage{float}
\usepackage{amsfonts}
\usepackage{amssymb}
\usepackage{mathtools}
\usepackage{amsthm}
\usepackage{bm}
\usepackage{xurl}
\usepackage{ulem}
\usepackage{amsmath}
\usepackage{algorithm}
\usepackage{algpseudocode}
\usepackage{hyperref}       
\usepackage[capitalize,noabbrev,nameinlink]{cleveref}

\theoremstyle{plain}

\newtheorem{theorem}{Theorem}[section]
\newtheorem{lemma}[theorem]{Lemma}
\newtheorem{proposition}[theorem]{Proposition}
\newtheorem{corollary}[theorem]{Corollary}
\newtheorem{definition}[theorem]{Definition}

\newtheorem{assumption}[theorem]{Assumption}

\definecolor{myblue}{HTML}{4363d8}
\definecolor{myorange}{HTML}{f58231}
\definecolor{mygreen}{HTML}{3cb44b}

\newcommand{\DP}{{\mathrm{DP}}}

\newcommand{\test}{\mathrm{test}}
\renewcommand{\P}{\mathbb{P}}
\newcommand{\E}{\mathbb{E}}
\newcommand{\R}{\mathbb{R}}
\newcommand{\CI}{\mathrm{CI}}
\newcommand{\algo}{\mathcal{A}}
\newcommand{\exch}{\mathrm{exch}}
\newcommand{\Mid}{\,\Vert\,}
\renewcommand{\d}{\mathrm{d}}

\makeatletter
\AddToHook{cmd/appendix/before}{\def\cref@section@alias{appendix}}
\makeatother

\makeatletter
\newcommand*{\transpose}{%
  {\mathpalette\@transpose{}}%
}
\newcommand*{\@transpose}[2]{%
  \raisebox{\depth}{$\m@th#1\intercal$}%
}

\title{Differentially Private E-Values}

\author{%
  Daniel Csillag \\ 
  FGV EMAp\\
  \texttt{daniel.csillag@fgv.br} \\
  \And
  Diego Mesquita \\
  FGV EMAp\\
  \texttt{diego.mesquita@fgv.br} \\
}

\begin{document}

\maketitle

\begin{abstract}
    E-values have gained prominence as flexible tools for statistical inference and risk control, enabling anytime- and post-hoc-valid procedures under minimal assumptions. However, many real-world applications fundamentally rely on sensitive data, which can be leaked through e-values. To ensure their safe release, we propose a general framework to transform non-private e-values into differentially private ones. Towards this end, we develop a novel biased multiplicative noise mechanism that ensures our e-values remain statistically valid. We show that our differentially private e-values attain strong statistical power, and are asymptotically as powerful as their non-private counterparts. Experiments across online risk monitoring, private healthcare, and conformal e-prediction demonstrate our approach's effectiveness and illustrate its broad applicability. 
\end{abstract}

\section{Introduction}

E-values have emerged as versatile tools for statistical inference and risk control, offering anytime- and post-hoc-valid guarantees under minimal assumptions. They are the backbone of a growing body of methods for continuous risk monitoring \citep{evalue-risk-monitoring}, change-point detection \citep{evalue-changepoint-1,evalue-changepoint-2}, test-time adaptation \citep{evalue-test-time-adaptation}, uncertainty quantification \citep{ecp-prior,ecp-gauthier}, and interpretability \citep{evalue-semantic}, among other inferential tasks.
However, many of the domains where these methods are most impactful involve sensitive data.
Applying existing e-value-based procedures directly to such data can compromise individual privacy, since their guarantees ensure statistical validity but not protection against information leakage.


Differential privacy provides a principled framework for addressing such challenges.
However, standard differential privacy mechanisms are insufficient for our setting: though privacy is attained, the resulting quantities are generally not statistically valid e-values.
In this paper, we resolve this by introducing novel biased multiplicative noise mechanisms, which ensure privacy while retaining statistical validity.
In this way, we are able to to convert any non-private e-value into a differentially private one.
To the best of our knowledge, this is the first instance of valid e-values that satisfy differential privacy.

Beyond establishing validity, we also exactly quantify the statistical power of our differentially private e-values, showing that it differs from that of the non-private e-value only by a factor that decreases with the number of samples, related to the mechanism's bias.
Crucially, as the number of samples goes towards infinity, we recover the power of the non-private test.
We also show that our private e-values inherit many of the usual properties of e-values, with no loss of privacy.

Our framework is general: the resulting differentially private e-values retain the compositional and optional continuation properties of standard e-values while ensuring privacy, and can be applied not only to hypothesis testing but also to confidence intervals and general e-value-based inference procedures.
We demonstrate the effectiveness of our approach through experiments in three real-world tasks: e-value based confidence intervals for the prevalence of diabetes, continuous monitoring of the risk of a deployed model, and e-conformal prediction for detection of phishing attacks.

\paragraph{Our contributions}
\begin{itemize}
    \item
        We propose the first general framework to privatize e-values, simultaneously satisfying differential privacy and statistical validity.
        Our framework works by introducing novel biased multiplicative noise mechanisms, which we show to be necessary to ensure validity of the resulting e-values.
    \item
        We derive an exact characterization of the statistical power of our differentially private e-values, in terms of their growth rates.
        As the number of samples grows, so does the growth rates, matching that of the non-private e-value at the limit.
    \item
        We prove that beyond basic validity, our differentially private e-values satisfy many compositional properties of e-values while preserving privacy, including optional continuation, e-to-p conversion and averaging (the latter under some restrictions);
        this enables the seamless use of many existing e-value-based procedures.
    \item
        \looseness=-1
        We demonstrate our framework on three real-world settings
        spanning private healthcare, online risk monitoring atop private data, and private predictive modelling.
        For each instance we derive effective procedures that are readily applicable by practitioners.
\end{itemize}

\noindent\textbf{Related work.}
The broader idea of statistical inference with differential privacy has already been the target of much attention. For example, \citep{prior-test-1,prior-test-2,prior-test-3} propose differentially private versions of classical tests, and \citep{prior-ci-1,prior-ci-2,prior-ci-3} propose private confidence intervals.
More recently, \citep{tot,pena-barrientos} proposed general frameworks from which differentially private hypothesis tests can be obtained.
However, thse are either inapplicable to e-values or discard essential properties of the e-values --- e.g., post-hoc validity, optional continuation, and advantages for multiple testing.
Our approach, in contrast, naturally benefits from all the usual properties of e-values, while also being significantly more data efficient (cf. Appendix~C.1).

\section{Background}

\subsection{Differential privacy}

Differential privacy (DP) is a framework for controlling the privacy loss incurred when releasing information about a dataset.
Introduced by \citep{dp-og}, differential privacy works by incorporating controlled noise to the processing of the dataset, so as to guarantee that the inclusion or exclusion of any single individual's data cannot be inferred from the outputs.
The first fully formalized notion of DP was $\epsilon$-differential privacy, which was later generalized to $(\epsilon, \delta)$-differential privacy.

\begin{definition}[$(\epsilon, \delta)$-differential privacy]
    A randomized algorithm $\algo(\cdot)$ satisfies $(\epsilon, \delta)$-differential privacy if, for all (fixed) neighboring datasets $D$ and $D'$ (i.e., differing by one record), for any set $A$,
    \[ \P[\algo(D) \in A] \leq e^\epsilon \P[\algo(D') \in A] + \delta. \]
    If $\delta=0$, algorithm $\algo(\cdot)$ satisfies $\epsilon$-differential privacy.
\end{definition}

Though influential, $(\epsilon, \delta)$-differential privacy is known to struggle with the composition of many differentially private procedures.
To improve on this, \citep{rdp} introduced $(\alpha, \epsilon)$-Rényi differential privacy, which generalizes $\epsilon$-differential privacy:

\begin{definition}[$(\alpha, \epsilon)$-Rényi differential privacy]
    A randomized algorithm $\algo(\cdot)$ satisfies $(\alpha, \epsilon)$-Rényi differential privacy for $\alpha > 1$ if, for all (fixed) neighboring datasets $D$ and $D'$ (i.e., differing by one record),
    \[
        D_\alpha \bigl( \algo(D) \Mid \algo(D') \bigr) \coloneq
        \frac{1}{\alpha-1} \log \E_{z \sim \algo(D')}\left[ \left( \frac{\d \algo(D)}{\d \algo(D')}(z) \right)^\alpha \right]
        \leq \epsilon,
    \]
    where $D_\alpha (P \Mid Q)$ is the $\alpha$-Rényi divergence between distributions $P$ and $Q$.
\end{definition}

It is worth mentioning that $(\alpha, \epsilon)$-Rényi differential privacy is closely related to classical $(\epsilon, \delta)$-differential privacy.
When $\alpha \to \infty$ we recover $\epsilon$-differential privacy \citep{rdp},
and satisfying $(\alpha, \epsilon)$-Rényi differential privacy also implies $(\epsilon', \delta)$-differential privacy with $\epsilon' = \epsilon + \log(1/\delta)/(\alpha-1)$, for any $\delta > 0$. \looseness=-1

Crucially, differential privacy benefits from two key properties: composition and post-processing.
Composition states that differentially private algorithms can be combined into new differentially private algorithms, whereas post-processing states that any post-processing of a differentially private algorithm retains its privacy.

\begin{proposition}[Composition of $(\alpha, \epsilon)$-Rényi differential privacy]\label{thm:background-composition}
    For any $\alpha > 1$,
    let algorithms $\algo_1(\cdot)$ and $\algo_2(\cdot)$ be $(\alpha, \epsilon_1)$- and $(\alpha, \epsilon_2)$-Rényi differentially private.
    Then the algorithm $\algo(\cdot) = (\algo_1(\cdot), \algo_2(\cdot))$ is $(\alpha, \epsilon_1 + \epsilon_2)$-Rényi differentially private.
\end{proposition}

\begin{proposition}[Post-processing of $(\alpha, \epsilon)$-Rényi differential privacy]\label{thm:background-postprocessing}
    For any $\alpha > 1$,
    let $\algo(\cdot)$ be a $(\alpha, \epsilon)$-Rényi differentially private algorithm.
    Then, for any post-processing operation $f$, the algorithm $\algo_f(\cdot) = f(\algo(\cdot))$ is also $(\alpha, \epsilon)$-Rényi differentially private.
\end{proposition}

Throughout the paper, we will refer primarily to $(\alpha, \epsilon)$-Rényi differential privacy. Nevertheless, our framework is just as applicable to standard $(\epsilon, \delta)$-differential privacy; see Appendix~B.1.

\subsection{E-values}

Consider the problem of testing a null hypothesis $H_0$ with data $[Y_1, Y_2, \ldots, Y_n] \eqcolon D$.
To this end, e-values serve as alternatives to the classic p-values as measures of evidence against the null.
Formally, an e-value $E(D)$ for the null hypothesis $H_0$ is a nonnegative random variable whose expectation is at most one under the null.

\begin{definition}[E-value]\label{def:e-value}
    A nonnegative real random variable $E(D)$ is an e-value for a null hypothesis $H_0$ if $\E[E(D)] \leq 1$ under $H_0$.
\end{definition}

A random variable is commonly said to be a ``valid'' e-value when it satisfies Definition~\ref{def:e-value}.

Any e-value can be converted to a p-value by simply taking its reciprocal, and any p-value can be converted to an e-value by a process termed calibration \citep{evalue-calibration}, albeit at a slight loss of power.
In a sense, e-values can be seen as p-values with richer structure \citep{evalue-posthoc},
in particular satisfying post-hoc validity, optional continuation, and merging through averaging.

\begin{proposition}[E-to-p conversion; post-hoc validity]\label{thm:background-e-to-p}
    If $E(D)$ is an e-value for a null hypothesis $H_0$, then $1/E(D)$ is a p-value for $H_0$. Moreover, it is a \emph{post-hoc valid} p-value \citep{evalue-posthoc}, i.e., a p-value that allows for an arbitrarily data dependent significance level $\alpha$, and all post-hoc valid p-values can be written as the reciprocal of some e-value.
\end{proposition}

\begin{proposition}[Optional continuation]\label{thm:background-optional-continuation}
    If $E_1(D_1)$ and $E_2(D_2)$ are e-values for a null hypothesis $H_0$ over independent\footnote{This can be refined to require only a sequential structure of the data, rather than full independence. We keep to independence for simplicity.} datasets $D_1$ and $D_2$, then $E_1(D_1) \cdot E_2(D_2)$ is also an e-value for $H_0$.
\end{proposition}

\begin{proposition}[Averaging]\label{thm:background-averaging}
    If $E_1(D_1)$ and $E_2(D_2)$ are e-values for a null hypothesis $H_0$, then all convex combinations $\eta E_1(D_1) + (1-\eta) E_2(D_2)$ for $\eta \in [0, 1]$ are e-values for $H_0$, regardless of any dependence between $D_1$ and $D_2$.
\end{proposition}

Defined here for hypothesis testing, e-values then serve as the building blocks of many higher-level procedures, from parameter inference \citep{savi,evalue-ppi,evalue-asympt-ppi}, risk monitoring \citep{evalue-risk-monitoring}, change-point detection \citep{evalue-changepoint-1,evalue-changepoint-2}, test-time adaptation \citep{evalue-test-time-adaptation}, uncertainty quantification \citep{ecp-prior,ecp-gauthier,ecp-gauthier-2}, interpretability \citep{evalue-semantic}, and more.

\section{Differentially Private E-Values}

In this section we present our framework for differentially private e-values, which guarantees privacy while retaining the validity of the resulting e-values.
We will first construct our differentially private e-values in the context of hypothesis testing, and then show how these can be used for general e-value-based procdures and confidence intervals.

\subsection{Hypothesis testing} \label{sec:method-hyp-test}

Suppose we have a (non-private) e-value $E(D)$ for a null hypothesis $H_0$, i.e., a nonnegative real random variable whose expectation over the data $D$ is at most 1 under the null hypothesis.
Our goal is to leverage $E(D)$ to construct a new e-value $E^\DP(D)$ satisfying $(\alpha, \epsilon)$-Rényi differential privacy.
We do this by introducing a measured amount of independent noise to the e-value, as such:
\begin{equation}\label{eq:e-dp-form}
    E^\DP(D) \coloneq E(D) \cdot e^{-\xi},
\end{equation}
with random noise $\xi$ independent from $E(D)$.
By passing the noise variable $\xi$ through the exponential and incorporating it multiplicatively we ensure that $E^\DP(D)$ remains nonnegative.

With well-designed choices for the distribution of $\xi$, we can simultaneously ensure differential privacy and validity.
First, note that since $\xi$ is independent from $E(D)$,
\[ \E[E^\DP(D)] = \E[E(D) \cdot e^{-\xi}] = \E[E(D)] \cdot \E[e^{-\xi}]; \]
and so, as long as the moment-generating function $t \mapsto \E[E^{t \xi}]$ exists and is at most $1$ at $t=-1$, we have that $E^\DP(D)$ will be a valid e-value: under the null,
\[ \E[E^\DP(D)] = \E[E(D)] \cdot \E[e^{-\xi}] \leq \E[E(D)] \leq 1, \]
where the last step follows from the fact that $E(D)$ is an e-value.

To obtain differential privacy, we then appeal to the post-processing theorem: $E^\DP(D)$ is $(\alpha, \epsilon)$-Rényi differentially private iff $\log E^\DP(D) = \log E(D) - \xi$ is $(\alpha, \epsilon)$-Rényi differentially private.
This allows us to leverage existing additive noise mechanisms such as Gaussian and Laplace mechanisms to attain differential privacy, as a function of the \emph{log-sensitivity} of the e-value, defined as
\begin{equation}\label{eq:sensitivity}
    \Delta_{\log} (E) \coloneq \sup_{\lvert D \Delta D' \rvert \leq 1} \lvert \log E(D) - \log E(D') \rvert,
\end{equation}
with $\lvert D \Delta D' \rvert \leq 1$ denoting that $D$ and $D'$ differ by a single record.

To ensure validity of the resulting e-values, we bias $\xi$ so as to ensure its moment-generating function satisfies the required bound.
This leads to the following biased Gaussian and Laplace mechanisms:

\begin{theorem}[Biased Gaussian mechanism]\label{thm:hyp-test-valid-gaussian}
    For any $\alpha > 1$ and $\epsilon > 0$, let
    \[ E^\DP(D) = E(D) \cdot e^{-\xi}, \qquad \xi \sim \mathcal{N}\biggl( \frac{\alpha [\Delta_{\log}(E)]^2}{4 \epsilon}, \frac{\alpha [\Delta_{\log}(E)]^2}{2 \epsilon} \biggr). \]
    Then $E^\DP(D)$ is a valid e-value satisfying $(\alpha, \epsilon)$-Rényi differential privacy.
\end{theorem}

\begin{theorem}[Biased Laplace mechanism]\label{thm:hyp-test-valid-laplace}
    For any $\alpha > 1$ and $\epsilon > 0$, let
    \[ E^\DP(D) = E(D) \cdot e^{-\xi}, \qquad \xi \sim \mathrm{Laplace}\bigl( -\log (1 - b_{\alpha,\epsilon}^2), b_{\alpha,\epsilon} \bigr), \]
    where
    \begin{align*}
        b_{\alpha,\epsilon} &\coloneq 1 \,/\, \mathbf{h}_\alpha^{-1}\left( (2\alpha - 1) e^{(\alpha - 1) \epsilon} \right),
        \\
        \mathbf{h}_\alpha(t) &\coloneq \alpha e^{(\alpha-1) \Delta_{\log}(E) t} + (\alpha-1) e^{-\alpha \Delta_{\log}(E) t} \qquad \text{for }t \geq 0.
    \end{align*}
    Then, as long as $b_{\alpha,\epsilon} < 1$,
    $E^\DP(D)$ is a valid e-value satisfying $(\alpha, \epsilon)$-Rényi differential privacy.
\end{theorem}

Interestingly, the biased Laplace mechanism is only viable when the sensitivity is not too high, due to the $b_{\alpha,\epsilon} < 1$ requirement. This is in contrast to the more usual (non-biased) Laplace mechanism for differential privacy, which is always applicable (but does not ensure validity of the e-values). The biased Gaussian mechanism, however, is always available.

We can also exactly characterize the stastical power of our differentially private e-values, in terms of their expected growth rates \citep{kelly}:

\begin{proposition}\label{thm:hyp-test-power}
    Let $E^\DP(D)$ be as in Equation~\ref{eq:e-dp-form}. Then
    \[ \E\left[\frac{1}{n} \log E^\DP(D)\right] = \E\left[\frac{1}{n} \log E(D)\right] - \frac{\E[\xi]}{n}. \]
\end{proposition}

Crucially, since the sensitivity of the log of e-values is usually at most constant w.r.t. $n$ (cf. Section~\ref{sec:experiments}), the penalty $\E[\xi]/n$ decays with a fast rate of $O(1/n)$. At the limit, our differentially private e-values are as powerful as their non-private counterparts.

Our differentially private e-values also behave in many of the usual ways, while preserving privacy.

\begin{proposition}[Optional continuation]\label{thm:optional-continuation-dp}
    If $E^\DP_1(D_1)$ and $E^\DP_2(D_2)$ are $(\alpha, \epsilon)$-Rényi differentially private e-values for a null hypothesis $H_0$, with data $D_1$ independent from $D_2$, then $E^\DP_1(D_1) \cdot E^\DP_2(D_2)$ is also an $(\alpha, \epsilon)$-Rényi differentially private e-value for $H_0$. Moreover, the release of both $E^\DP_1(D_1)$ and $E^\DP_1(D_1) \cdot E^\DP_2(D_2)$ is also $(\alpha, \epsilon)$-Rényi differentially private.
\end{proposition}
\begin{proposition}[E-to-p conversion]\label{thm:p-to-e-dp}
    If $E^\DP(D)$ is an $(\alpha, \epsilon)$-Rényi differentially private e-value for a null hypothesis $H_0$, then $1/E^\DP(D)$ is an $(\alpha, \epsilon)$-Rényi differentially private post-hoc valid p-value for $H_0$.
\end{proposition}

A notable exception, however, is averaging. Though the average of e-evalues is always an e-value, the privacy guarantee may degrade under certain conditions.
\begin{proposition}[Independent Averaging] \label{thm:independent-averaging-dp}
    If $E^\DP_1(D_1)$ and $E^\DP_2(D_2)$ are $(\alpha, \epsilon)$-Rényi differentially private e-values for a null hypothesis $H_0$, then for any $\eta \in [0, 1]$, $\eta E^\DP_1(D_1) + (1-\eta) E^\DP_2(D_2)$ is also an $(\alpha, \epsilon)$-differentially private e-value for $H_0$.
\end{proposition}
\begin{proposition}[Dependent Averaging] \label{thm:dependent-averaging-dp}
    More generally, if $E^\DP_1(D)$ and $E^\DP_2(D)$ are $(\alpha, \epsilon)$-Rényi differentially private e-values for a null hypothesis $H_0$, then for any $\eta \in [0, 1]$, $\eta E^\DP_1(D) + (1-\eta) E^\DP_2(D)$ is an $(\alpha, 2\epsilon)$-differentially private e-value for $H_0$.
\end{proposition}

\subsection{Algorithms atop e-values and confidence intervals}

Beyond hypothesis testing, e-values are also commonly used as fundamental building blocks in larger algorithms.
As long as the number of e-values used is finite, the standard composition theorems of differential privacy apply, and so for appropriately chosen values of $(\alpha, \epsilon)$ we can simply replace the procedure's e-values with our differentially private ones and attain validity with privacy.

Formally, we have an algorithm $\algo(E_1, \ldots, E_k)$, receiving as input $k$ e-values for $k$ respective null hypotheses $H^{(k)}_0$.
We need some notion of validity of the overall algorithm, which we assume holds whenever the input e-values are all valid.
\begin{assumption}\label{thm:algo-valid-assumption}
    If $E_1, \ldots, E_k$ are valid e-values for the nulls $H^{(1)}, \ldots, H^{(k)}$, then $\algo(E_1, \ldots, E_k)$ is valid.
\end{assumption}

Then the next result follows by the standard composition theorem of (Rényi) differential privacy.

\begin{theorem}\label{thm:general-algo-valid}
    Under Assumption~\ref{thm:algo-valid-assumption},
    let $\alpha > 1$ and $\epsilon > 0$.
    For each $j = 1, \ldots, k$, let $E^\DP_j$ be an $(\alpha, \epsilon/k)$-Rényi differentially private e-value for the null $H_0^{(j)}$ (e.g., obtained through Theorems~\ref{thm:hyp-test-valid-gaussian} and \ref{thm:hyp-test-valid-laplace}).
    Then $\algo(E^\DP_1, \ldots, E^\DP_k)$ is valid and $(\alpha, \epsilon)$-Rényi differentially private.
\end{theorem}

This works out of the box for many algorithms. However, it is not enough for algorithms that (formally) depend on an infinite set of e-values. A particularly notable example of this are confidence intervals; an e-value-based confidence interval for some parameter $\theta^\star \in \Theta$ is typically defined by taking a family of e-values $(E_\theta(D))_{\theta \in \Theta}$ for nulls $H^{(\theta)}_0 : \theta^\star = \theta$, and inverting the test as $C_\alpha(D) \coloneq \{ \theta \in \Theta : E_\theta(D) < 1/\alpha \}$.

To resolve this, we define a procedure that leverages a finite amount of our differentially private e-values to provide a provably valid confidence interval, under the assumption that the log of the (non-private) e-value, $\log E_\theta(D)$, is locally Lipschitz in the parameter $\theta$. For simplicity we present here the scalar case where $\Theta \subset \R$, but the same technique can be applied in $\R^d$:
Let $0 = a_0 < \cdots < a_k = 1$ be a partition of $[0, 1]$, and define the corresponding midpoints $\theta_{j} \coloneq (a_{j-1} + a_j)/2$.
Because $\log E_\theta(D)$ is locally Lipschitz, it is $L_j$-Lipschitz within $[a_{j-1}, a_j]$, for each $j = 1, \ldots, k$; and thus, for all $\theta' \in [a_{j-1}, a_j]$,
\begin{align}
    \log E_{\theta'}(D) &\geq \log E_{\theta_j}(D) - L_j \lvert \theta' - \theta_j \rvert \nonumber
    \\
    \implies
    E_{\theta'}(D) &\geq E_{\theta_j}(D) \cdot e^{-L_j \lvert \theta' - \theta_j \rvert} \label{eq:e-tilde}
    \\ &\geq E_{\theta_j}(D) \cdot e^{-L_j (a_j - a_{j-1})} \eqcolon \widetilde{E}_{\theta_j}(D), \nonumber
\end{align}
and the last right-hand-side is independent of $\theta'$.
So $\widetilde{E}_{\theta_j}(D)$ is simultaneously an e-value for all nulls $H_0 : \theta^\star = \theta'$, $\theta' \in [a_{j-1}, a_j]$, as
\[ \E[E_{\theta_j}(D) \cdot e^{-L_j (a_j - a_{j-1})}] \leq \E[E_{\theta'}(D)] \leq 1. \]
Now that we have a finite amount of e-values covering all our nulls, we can straightforwardly define a procedure atop them to generate an actual CI from the finite e-values:
\begin{equation*}\label{eq:ci}
    \CI_{\alpha}\!\left( \widetilde{E}_{\theta_1}(D), \ldots, \widetilde{E}_{\theta_k}(D) \right) \coloneq \!\! \bigcup_{\substack{j=1 \\ \widetilde{E}_{\theta_j}(D) \leq 1/\alpha }}^k [a_{j-1}, a_j].
\end{equation*}
Differential privacy is now easily attained by creating differentially private versions of $\widetilde{E}_{\theta_1}(D), \ldots, \widetilde{E}_{\theta_k}(D)$ (e.g., by Theorems~\ref{thm:hyp-test-valid-laplace} and \ref{thm:hyp-test-valid-gaussian}) and applying Theorem~\ref{thm:general-algo-valid}.
\begin{corollary}\label{thm:valid-ci}
    Suppose $\log E_\theta(D)$ is locally Lipschitz in $\theta$.
    Let $\alpha > 1$ and $\epsilon > 0$,
    and let $\widetilde{E}^\DP_{\theta_1}, \ldots, \widetilde{E}^\DP_{\theta_k}$ be $(\alpha, \epsilon/k)$-Rényi differentially private versions of $\widetilde{E}_{\theta_1}, \ldots, \widetilde{E}_{\theta_k}$ as defined in Equation~\ref{eq:e-tilde}.
    Then $\CI_\alpha\bigl( \widetilde{E}^\DP_{\theta_1}(D), \ldots, \widetilde{E}^\DP_{\theta_k}(D) \bigr)$ is an $(\alpha, \epsilon)$-Rényi differentially private confidence interval for $\theta^\star$, i.e., it satisfies $(\alpha, \epsilon)$-Rényi differential privacy, and
    \[ \P\left[ \theta^\star \in \CI_\alpha\!\left(\widetilde{E}^\DP_{\theta_1}(D), \ldots, \widetilde{E}^\DP_{\theta_k}(D)\right) \right] \geq 1 - \alpha. \]
\end{corollary}

\section{Experiments and Applications}\label{sec:experiments}

In this section we empirically evaluate our method in three settings:
(i) confidence intervals for the prevalence of diabetes with private patient data (Section~\ref{sec:experiment-confseq});
(ii) private anytime-valid hypothesis testing for online risk monitoring (Section~\ref{sec:experiment-hyptest}); and
(iii) e-conformal prediction for the predictive modelling of online phishing attacks (Section~\ref{sec:experiment-ecp}).
Experiment details can be found in the supplementary material.

Code for all experiments can be found on \url{https://github.com/dccsillag/experiments-evalue-dp}.
All experiments were run on an AMD Ryzen 9 5950X CPU, with 64GB of RAM.
That said, they are lightweight and should easily run on weaker hardware.

\subsection{Private e-confidence intervals} \label{sec:experiment-confseq}

\begin{figure*}[t]
    \centering
    \includegraphics[width=\textwidth]{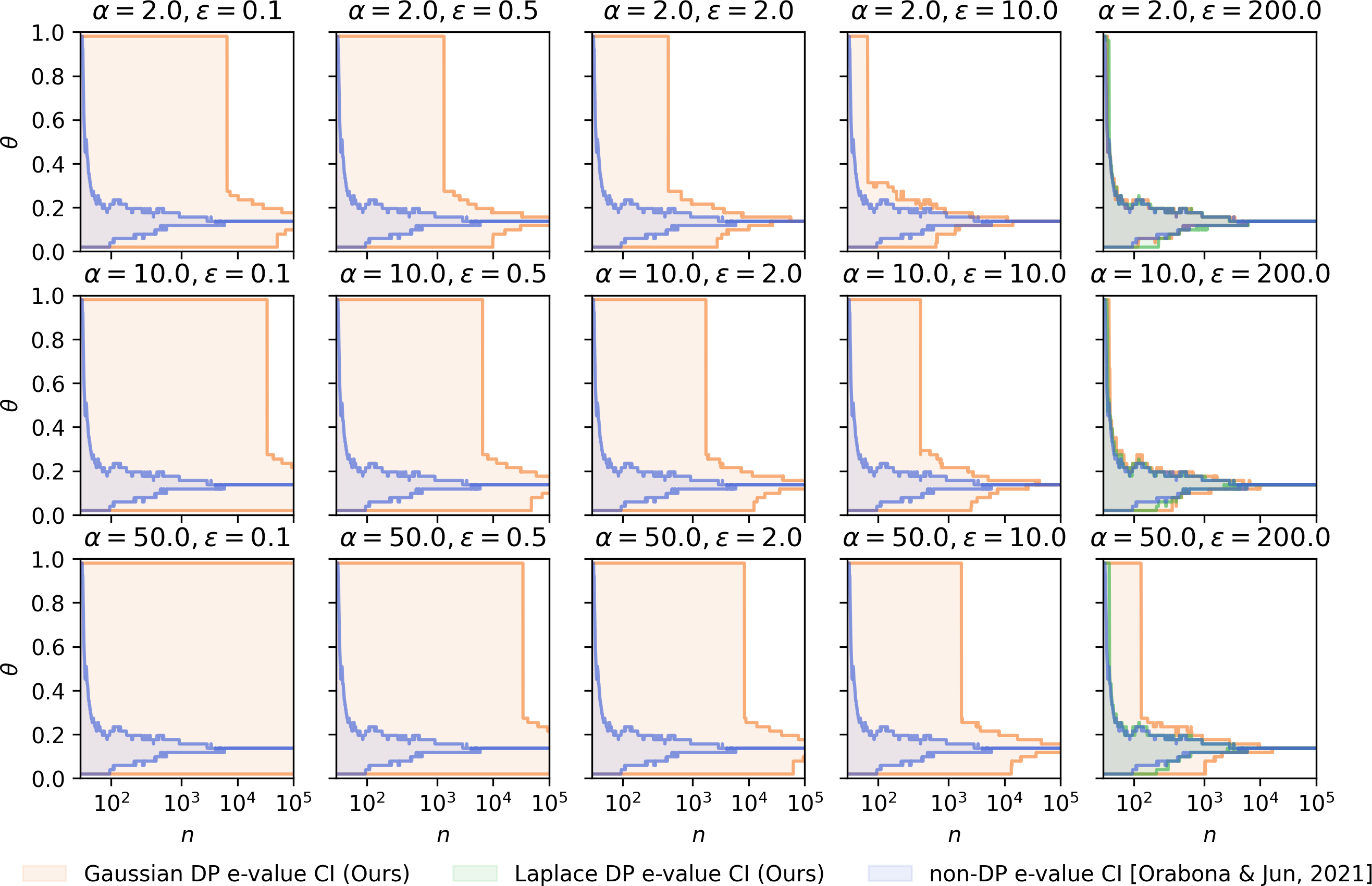}
    \caption{\textbf{Private e-confidence intervals.} We use our differentially private e-values to produce private confidence intervals for the prevalence of diabetes over a population. We illustrate the confidence intervals varying the number of samples $n$ (x axis) and Rényi privacy budgets $(\alpha, \epsilon)$.
    We display the non-private confidence intervals (\textcolor{myblue}{blue}) along with those obtained through our biased Gaussian mechanism (\textcolor{myorange}{orange}), and when applicable also the as the ones obtained via our biased Laplace mechanism (\textcolor{mygreen}{green}).
    Note that our biased Gaussian mechanism always converges to the non-private confidence intervals, as does the biased Laplace mechanism when available.
    }
    \label{fig:confseq}
\end{figure*}

We start by the problem of producing a confidence interval for the mean of a bounded random variable, following the work of \citep{evalue-mean}.
To do so, we leverage the e-value for the mean of \citep{evalue-mean} with betting following Cover's universal portfolios \citep{cover,evalue-mean-cover}:
\begin{equation} \label{eq:evalue-mean-cover}
    E_{\theta}(D) = \prod_{i=1}^n \bigl( 1 + \lambda_i (Y_i - \theta) \bigr)
    \quad \text{with} \quad
    \lambda_i = \frac{\E_{\lambda \sim F}\left[ \lambda \cdot \prod_{j=1}^i \bigl( 1 + \lambda (Y_j - \theta) \bigr) \right]}{\E_{\lambda \sim F}\left[ \prod_{j=1}^i \bigl( 1 + \lambda (Y_j - \theta) \bigr) \right]}, \nonumber
\end{equation}
where $F$ is a distribution with support in $[\lambda_{\inf}, \lambda_{\sup}] \subset (-1/(1-\theta), 1/\theta)$.

\begin{figure*}[t]
    \centering
    \includegraphics[width=\textwidth]{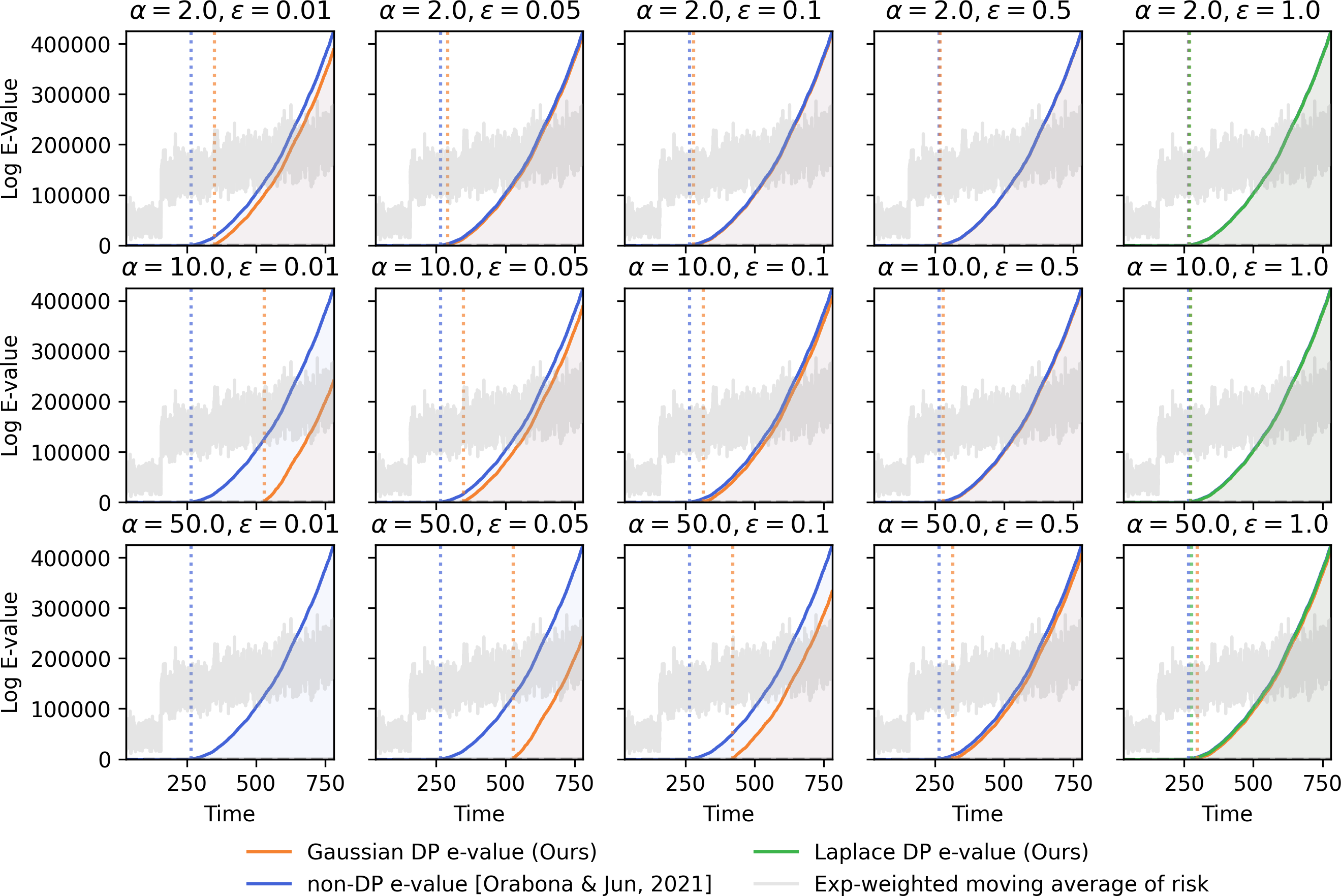}
    \caption{\textbf{Anytime-valid private hypothesis testing.} We apply our differentially private e-values for private continuous risk monitoring by the means of anytime-valid hypothesis testing. Our goal is to detect significant increases in the risk of a deployed predictive model, whose moving average can be seen in the background. The curves indicate the log of the e-values for the null hypothesis that the risk is under a certain safety threshold. Our private e-values closely match the non-private ones, quickly rejecting the null hypothesis after the change-point even for very small values of $\epsilon$.}
    \label{fig:hyptest}
\end{figure*}

Besides great performance, this choice of betting strategy $(\lambda_i)_{i=1}^n$ has the key property that the e-values $E_{\theta}(D)$ are invariant to permutations of the data $Y_1, \ldots, Y_n$ \citep{evalue-mean-cover,ville}.
Straightfoward computation then yields the following bound on the log--sensitivity of $E_{\theta}(D)$:
\begin{proposition}\label{thm:sensitivity-evalue-mean}
    Let $F$ have support in $[\lambda_{\inf}, \lambda_{\sup}] \subset (-1/(1-\theta), 1/\theta)$.
    The log-sensitivity of $E_\theta(D)$, $\Delta_{\log}(E_{\theta})$, is upper bounded as
    \begin{align*}
        \Delta_{\log}(E_{\theta}) &\leq
            \max\bigl\{
            \log ( 1 + \max\{\lambda_{\sup} (1-\theta), -\lambda_{\inf} \theta\} ),
            \\ &\qquad\quad -\log ( 1 + \min\{\lambda_{\inf}(1-\theta),-\lambda_{\sup}\theta\} )
            \bigr\}.
    \end{align*}
\end{proposition}
Moreover, to apply our Proposition~\ref{thm:valid-ci} for constructing valid confidence intervals, we further bound the Lipschitz constant of the log of the e-value over $\theta \in [\theta_{\inf}, \theta_{\sup}]$:
\begin{proposition}\label{thm:lipschitz-bound-mean}
    For distributions $F$ with support contained in $[\lambda_{\inf}, \lambda_{\sup}]$, the Lipschitz constant of $\log E_\theta(D)$ over $\theta \in [\theta_{\inf}, \theta_{\sup}]$ is upper bounded by
    \[
        \lVert\theta \mapsto \log E_\theta(D)\rVert_{\mathrm{Lip}}
        \leq \max\left\{
            \left\lvert \frac{\lambda_{\sup}}{1 - \lambda_{\sup} \theta_{\sup}} \right\rvert,
            \left\lvert \frac{\lambda_{\inf}}{1 + \lambda_{\inf} (1 - \theta_{\inf})} \right\rvert
            \right\}.
    \]
\end{proposition}

Altogether, we can now directly use the biased Gaussian and Laplace mechanisms defined in Section~\ref{sec:method-hyp-test} along with Proposition~\ref{thm:valid-ci} to construct differentially private e-value-based confidence intervals for our data.

For our experiment, we use the e-values from Equation~\ref{eq:evalue-mean-cover} atop the dataset of \citep{dataset-diabetes} for estimating the prevalence of diabetes,
and vary the Rényi privacy budget over $\alpha \in \{2, 10, 50\}$ and $\epsilon \in \{0.1, 0.5, 2, 10, 200\}$.
The results can be seen in Figure~\ref{fig:confseq};
we note that our method with biased Gaussian noise always converges to the non-private intervals as $n \to \infty$, and very closely matches the non-private confidence intervals over all $n$ with $\epsilon = 10$.
For smaller values of $\epsilon$ (e.g. 0.1), our method requires only mild increases in data size, past a certain point.
The biased Laplace mechanism, however, is not defined for most combinations of $(\alpha, \epsilon)$, due to the $b_{\alpha,\epsilon} < 1$ requirement (cf. Theorem~\ref{thm:hyp-test-valid-laplace}). But when it is defined, it is quite accurate.

\subsection{Anytime-valid hypothesis testing for online risk monitoring} \label{sec:experiment-hyptest}

We now turn our attention to the problem of online risk monitoring with private data.
In this setting we have a pre-trained predictive model $\widehat{\mu}$, and want to continuously track its test loss so as to ensure it does not go over some predetermined safety threshold.
Inspired by the work of \citep{evalue-risk-monitoring} and \citep{evalue-ppi},
we frame this as an anytime-valid hypothesis test for the null hypothesis $H_0 : \E[\mathrm{Loss}(\widehat{\mu}(X_i), Y_i)] \leq \mathrm{SafetyThreshold} \text{ for all } i$. This corresponds to a one-sided test for the mean, which we can perform using the same e-value from Equation~\ref{eq:evalue-mean-cover}, but with $F$ now being the uniform distribution with support in $[0, c/\theta)$ for some $0 < c < 1$.

\begin{figure*}[t]
    \centering
    \includegraphics[width=\textwidth]{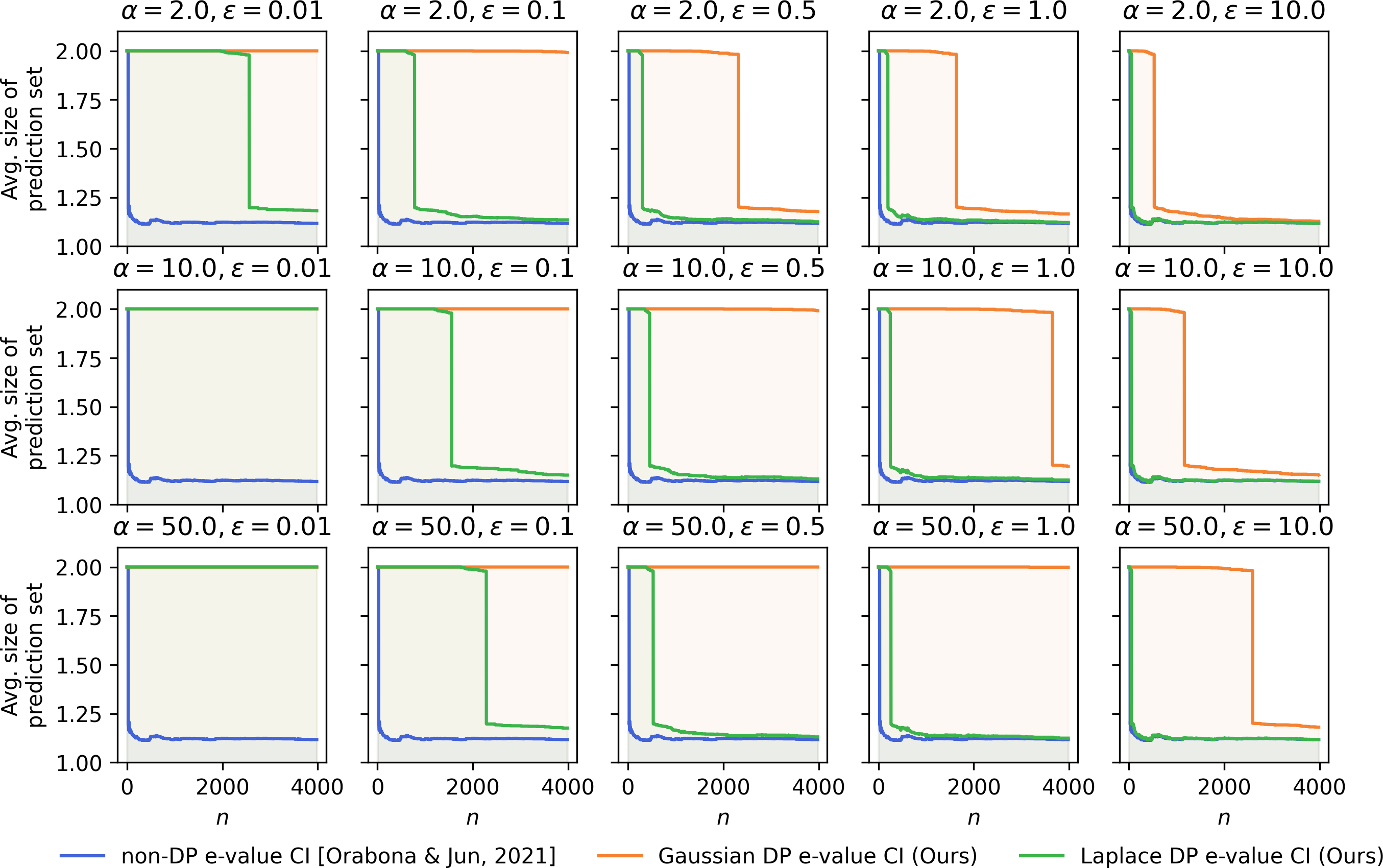}
    \caption{\textbf{Private e-conformal prediction.} We also apply our e-values to produce private predictive sets with post-hoc guarantees through e-conformal prediction. In this setting the biased Laplace mechanism is generally applicable, and consistently outperforms the biased Gaussian mechanism, closely matching the average size of non-private prediction sets.}
    \label{fig:conformal}
\end{figure*}
Crucially, in to attain anytime-validity with differential privacy we conduct our inference in batches, and appeal to the private optional continuation property (Proposition~\ref{thm:optional-continuation-dp}).
In each batch we compute the privatized form of the e-value of Equation~\ref{eq:evalue-mean-cover} with the one-sided distribution $F$, and combine the e-values across batches through multiplication.

For this experiment we use the data of \citep{dataset-covertype} with a simulated change-point.
Figure~\ref{fig:hyptest} shows this procedure in action, with the Rényi privacy budget varying over $\alpha \in \{2, 10, 50\}$ and $\epsilon \in \{0.01, 0.05, 0.1, 0.5, 1\}$.
We note the remarkable accuracy of our biased Gaussian mechanism in this setting, even for small values of $\epsilon$, which often rejects the null hypothesis (and thus triggers an alarm of a change of distribution) very close to the non-private procedure, and thus also close to the actual change-point in the data.
However, similar to our previous example in Section~\ref{sec:experiment-confseq}, the biased Laplace mechanism is only defined for most higher values of $\epsilon$.

\subsection{Private e-conformal prediction} \label{sec:experiment-ecp}

Finally, we apply our method to privatize e-conformal prediction \citep{ecp-gauthier,ecp-prior}.
E-conformal prediction provides predictive uncertainty quantification for machine learning models, achieving strong guarantees (e.g. post-hoc validity) by leveraging e-values.
Given calibration samples $D = (X_i, Y_i)_{i=1}^n \subset \mathcal{X} \times \mathcal{Y}$ and a
a score function $s : \mathcal{X} \times \mathcal{Y} \to \R_{>0}$ (which typically incorporates a predictive model),
e-conformal prediction produces predictive sets as
\[ C_\alpha(x) \coloneq \left\{ y \in \mathcal{Y} : E^{\exch}(D; s(x, y)) < 1/\alpha \right\}, \]
where $E^{\exch}(D; s(x, y))$ is an e-value for the null that $s(X_1, Y_1), \ldots, s(X_n, Y_n), s(x, y)$ are exchangeable random variables.
A common choice is
\[ E^\exch(D; S^\test) = \frac{(n+1) S^\test}{\sum_{i=1}^n s(X_i, Y_i) + S^\test}. \]
We furthermore consider that the conformity score $s$ has its image contained within a finite set ${s_1, \ldots, s_K} \subset [a, b]$.
This makes it so that $C_\alpha(\cdot)$ is a function of a finite amount of e-values (one for each $s_i$, $i = 1, \ldots, K$) and affords us the following bound on the log-sensitivity:

\begin{proposition}\label{thm:sensitivity-exch}
    Suppose the score function and $S^\test$ are all contained in $[a, b]$.
    Then the log-sensitivity of $E^\exch(D; S^\test)$, $\Delta_{\log}(E^\exch)$, is upper bounded as
    \[ \Delta_{\log}(E^\exch) \leq 2 \cdot \frac{b/a}{n+1}. \]
\end{proposition}

Then, by privatizing these e-values for exchangeability and using Theorem~\ref{thm:general-algo-valid}, we obtain private e-conformal predictive sets.

Figure~\ref{fig:conformal}
shows the predictive set sizes for conformal prediction atop the phishing classification dataset of \citep{dataset-phishing}, for Rényi privacy budgets with $\alpha \in \{2, 10, 50\}$ and $\epsilon \in \{0.01, 0.1, 0.5, 1, 10\}$.
For our conformity score, we use a truncated version of the score proposed by \citep{ecp-gauthier}, quantized onto 500 bins in $[1, 100]$, which affords us reasonably low sensitivities.
Remarkably, we find that in this instance the biased Laplace is (i) always defined, and (ii) consistently outperforms the biased Gaussian mechanism, quickly approaching the non-private average predictive set size even for small values of $\epsilon$.

\section{Conclusion}

In this work we introduced the first general framework for constructing differentially private e-values that simultaneously guarantee privacy protection and statistical validity. Through novel biased multiplicative noise mechanisms, we showed that our private e-values retain strong statistical power -- asymptotically matching their non-private counterparts -- while preserving their key compositional properties. Our experiments on online risk monitoring, private healthcare analytics, and conformal prediction demonstrate the practical effectiveness of our approach across diverse privacy-sensitive applications. By bridging differential privacy and e-value-based inference, this work enables the principled deployment of flexible e-value-based statistical procedures on sensitive data, broadening the applicability of modern inference methods to domains with rigorous privacy requirements.

\subsubsection*{Acknowledgements}

This work was supported by the Fundação Carlos Chagas Filho de Amparo à Pesquisa do Estado do Rio de Janeiro FAPERJ (SEI-260003/000709/2023) and the Conselho Nacional de Desenvolvimento Científico e Tecnológico CNPq (404336/2023-0, 305692/2025-9).

\bibliographystyle{apalike}
\bibliography{bibliography}

\appendix

\section{Proofs from the main text}

For proofs of Propositions~2.3 and 2.4, see \citep{rdp}.
The proof of Proposition~2.6 can be found in \citep{evalue-posthoc}. For Propositions~2.7 and 2.8, see \citep{evalue-calibration}.

\begin{theorem}[Biased Gaussian mechanism for a single e-value; Theorem~3.1 in the main text]
    For any $\alpha > 1$ and $\epsilon > 0$, let
    \[ E^\DP(D) = E(D) \cdot e^{-\xi}, \qquad \xi \sim \mathcal{N}\biggl( \frac{\alpha [\Delta_{\log}(E)]^2}{4 \epsilon}, \frac{\alpha [\Delta_{\log}(E)]^2}{2 \epsilon} \biggr). \]
    Then $E^\DP(D)$ is a valid e-value satisfying $(\alpha, \epsilon)$-Rényi differential privacy.
\end{theorem}

\begin{proof}
    We are considering $\xi \sim \mathcal{N}(\mu, \sigma^2)$ for some $\mu$ and $\sigma^2$.
    By the post-processing theorem,
    $E^\DP(D) = E(D) \cdot e^{-\xi}$ is $(\alpha, \epsilon)$-RDP iff $\log E^\DP(D) = \log \left( E(D) \cdot e^{-\xi} \right) = \log E(D) - \xi$ is $(\alpha, \epsilon)$-RDP; note that $(\log E(D) - xi) \sim \mathcal{N}(\log E(D) - \mu, \sigma^2)$. It then follows, using Proposition 7 of \citep{rdp}, that
    \begin{align*}
        D_\alpha (\log E^\DP(D) \Mid \log E^\DP(D'))
        &= D_\alpha (\mathcal{N}(\log E(D) - \mu, \sigma^2) \Mid \mathcal{N}(\log E(D') - \mu, \sigma^2))
        \\ &= D_\alpha (\mathcal{N}(0, \sigma^2) \Mid \mathcal{N}((\log E(D') - \mu) - (\log E(D) - \mu), \sigma^2))
        \\ &= D_\alpha (\mathcal{N}(0, \sigma^2) \Mid \mathcal{N}(\log E(D') - \log E(D), \sigma^2))
        \\ &= \frac{\alpha (\log E(D') - \log E(D))^2}{2 \sigma^2}.
    \end{align*}
    For this to be at most $\epsilon$, we require:
    \begin{align*}
        & \frac{\alpha (\log E(D') - \log E(D))^2}{2 \sigma^2} \leq \epsilon
        \\ \iff& \frac{\alpha (\log E(D') - \log E(D))^2}{2 \epsilon} \leq \sigma^2
        \\ \impliedby& \frac{\alpha [\Delta_{\log}(E)]^2}{2 \epsilon} \leq \sigma^2;
    \end{align*}
    So, taking $\sigma^2 = \alpha [\Delta_{\log}(E)]^2 / 2 \epsilon$ ensures $(\alpha, \epsilon)$-Renyi differential privacy.

    Thus, all that remains is to find the smallest bias $\mu$ that ensures validity of $E^\DP(D)$.
    As argued in Section~3.1, all we need is that the MGF of $\xi$ is at most one at $-1$. So it follows, using the form of the MGF of an univariate normal:
    \begin{align*}
        \E[e^{-\xi}] &= \exp\left( -\mu + \frac{\alpha [\Delta_{\log}(E)]^2 / 2 \epsilon}{2} \right) \leq 1
        \\ \iff& -\mu + \frac{\alpha [\Delta_{\log}(E)]^2 / 2 \epsilon}{2} \leq 0
        \\ \iff& \mu \geq \frac{\alpha [\Delta_{\log}(E)]^2 / 2 \epsilon}{2} = \frac{\alpha [\Delta_{\log}(E)]^2}{4 \epsilon}.
        \qedhere
    \end{align*}
\end{proof}

In order to derive the biased Laplace mechanism, we first introduce the following Lemma, which is a proper generalization of Proposition 6 of \citep{rdp}.

\begin{lemma} \label{thm:renyi-laplace}
    Let $\alpha > 1$. Then the $\alpha$-Rényi divergence between $\mathrm{Laplace}(0, b)$ and $\mathrm{Laplace}(\mu, b)$ for any $\mu \in \R$ and $b > 0$ is given by
    \[
        D_\alpha\bigl( \mathrm{Laplace}(0, b) \Mid \mathrm{Laplace}(\mu, b) \bigr) = \frac{1}{\alpha-1} \log \left( \frac{\alpha}{2\alpha-1} e^{(\alpha-1)|\mu|/b} + \frac{\alpha-1}{2\alpha-1} e^{-\alpha|\mu|/b} \right).
    \]
    Additionally, $D_\alpha\bigl( \mathrm{Laplace}(0, b) \Mid \mathrm{Laplace}(\mu, b) \bigr)$ is increasing in $\lvert \mu \rvert$.
\end{lemma}

\begin{proof}
    By definition of $D_\alpha$,
    \begin{align*}
        D_\alpha\bigl( \mathrm{Laplace}(0, b) \Mid \mathrm{Laplace}(\mu, b) \bigr)
        &= \frac{1}{\alpha-1} \log \int_{-\infty}^{\infty} \left(\frac{1}{2b} e^{-|x|/b}\right)^\alpha \left(\frac{1}{2b} e^{-|x-\mu|/b}\right)^{1-\alpha} \, dx
        \\ &= \frac{1}{\alpha-1} \log \int_{-\infty}^{\infty} \frac{1}{2b} e^{-\alpha|x|/b - (1-\alpha)|x-\mu|/b}.
    \end{align*}

    We need to evaluate
    \[
        I = \int_{-\infty}^{\infty} \frac{1}{2b} e^{-\alpha|x|/b - (1-\alpha)|x-\mu|/b} \, dx.
    \]

    First, let us assume that $\mu \geq 0$.
    We then split the integral into three regions:

    \textbf{Region 1:} $x < 0$
    \[
        \int_{-\infty}^0 \frac{1}{2b} e^{\alpha x/b - (1-\alpha)(\mu-x)/b} \, dx = \frac{1}{2b} e^{-(1-\alpha)\mu/b} \int_{-\infty}^0 e^{x/b} \, dx = \frac{1}{2} e^{-(1-\alpha)\mu/b}
    \]

    \textbf{Region 2:} $0 \leq x < \mu$
    \begin{align*}
        \int_0^\mu \frac{1}{2b} e^{-\alpha x/b - (1-\alpha)(\mu-x)/b} \, dx &= \frac{1}{2b} e^{-(1-\alpha)\mu/b} \int_0^\mu e^{-(2\alpha-1)x/b} \, dx \\
        &= \frac{1}{2} e^{-(1-\alpha)\mu/b} \frac{1 - e^{-(2\alpha-1)\mu/b}}{2\alpha-1}
    \end{align*}

    \textbf{Region 3:} $x \geq \mu$
    \[
        \int_\mu^\infty \frac{1}{2b} e^{-\alpha x/b - (1-\alpha)(x-\mu)/b} \, dx = \frac{1}{2b} e^{(1-\alpha)\mu/b} \int_\mu^\infty e^{-x/b} \, dx = \frac{1}{2} e^{-\alpha\mu/b}
    \]

    We thus have
    \begin{align*}
        I
        &= \frac{1}{2} e^{-(1-\alpha)\mu/b} \left[ 1 + \frac{1 - e^{-(2\alpha-1)\mu/b}}{2\alpha-1} \right] + \frac{1}{2} e^{-\alpha\mu/b}
        \\ &= \frac{1}{2} e^{(\alpha-1)\mu/b} \left[ 1 + \frac{1 - e^{-(2\alpha-1)\mu/b}}{2\alpha-1} \right] + \frac{1}{2} e^{-\alpha\mu/b}
        \\ &= \frac{1}{2} e^{(\alpha-1)\mu/b} \left[ \frac{2\alpha-1 + 1 - e^{-(2\alpha-1)\mu/b}}{2\alpha-1} \right] + \frac{1}{2} e^{-\alpha\mu/b}
        \\ &= \frac{1}{2} e^{(\alpha-1)\mu/b} \left[ \frac{2\alpha - e^{-(2\alpha-1)\mu/b}}{2\alpha-1} \right] + \frac{1}{2} e^{-\alpha\mu/b}
        \\ &= \frac{\alpha}{2\alpha-1} e^{(\alpha-1)\mu/b} - \frac{1}{2} e^{(\alpha-1)\mu/b} \frac{e^{-(2\alpha-1)\mu/b}}{2\alpha-1} + \frac{1}{2} e^{-\alpha\mu/b}
        \\ &= \frac{\alpha}{2\alpha-1} e^{(\alpha-1)\mu/b} - \frac{1}{2} \frac{e^{-\alpha\mu/b}}{2\alpha-1} + \frac{1}{2} e^{-\alpha\mu/b}
        \\ &= \frac{\alpha}{2\alpha-1} e^{(\alpha-1)\mu/b} + \frac{\alpha-1}{2\alpha-1} e^{-\alpha\mu/b}.
    \end{align*}

    For $\mu < 0$, we have
    \begin{align*}
        D_\alpha\bigl( \mathrm{Laplace}(0, b) \Mid \mathrm{Laplace}(\mu, b) \bigr)
        &= D_\alpha\bigl( -\mathrm{Laplace}(0, b) \Mid -\mathrm{Laplace}(\mu, b) \bigr)
        \\ &= D_\alpha\bigl( \mathrm{Laplace}(0, b) \Mid \mathrm{Laplace}(-\mu, b) \bigr)
        \\ &= D_\alpha\bigl( \mathrm{Laplace}(0, b) \Mid \mathrm{Laplace}(\lvert \mu \rvert, b) \bigr),
    \end{align*}
    from which we conclude the computation.

    Finally, to see that this is increasing in $\lvert \mu \rvert$, we take the derivative:
    \begin{align*}
        & \frac{\mathrm{d}}{\mathrm{d} \mu} \frac{1}{\alpha-1} \log \left( \frac{\alpha}{2\alpha-1} e^{(\alpha-1)\mu/b} + \frac{\alpha-1}{2\alpha-1} e^{-\alpha\mu/b} \right)
        \\ &= \frac{1}{\alpha-1} \left( \frac{\alpha}{2\alpha-1} e^{(\alpha-1)\mu/b} + \frac{\alpha-1}{2\alpha-1} e^{-\alpha\mu/b} \right)^{-1} \frac{\mathrm{d}}{\mathrm{d} \mu} \left( \frac{\alpha}{2\alpha-1} e^{(\alpha-1)\mu/b} + \frac{\alpha-1}{2\alpha-1} e^{-\alpha\mu/b} \right) > 0
        \\ &\iff \frac{\mathrm{d}}{\mathrm{d} \mu} \left( \frac{\alpha}{2\alpha-1} e^{(\alpha-1)\mu/b} + \frac{\alpha-1}{2\alpha-1} e^{-\alpha\mu/b} \right) > 0
        \\ &\iff \frac{\alpha}{2\alpha-1} \frac{\mathrm{d}}{\mathrm{d} \mu} e^{(\alpha-1)\mu/b} + \frac{\alpha-1}{2\alpha-1} \frac{\mathrm{d}}{\mathrm{d} \mu} e^{-\alpha\mu/b} > 0
        \\ &\iff \frac{\alpha}{2\alpha-1} e^{(\alpha-1)\mu/b} \frac{\mathrm{d}}{\mathrm{d} \mu} ((\alpha-1)\mu/b) + \frac{\alpha-1}{2\alpha-1} e^{-\alpha\mu/b} \frac{\mathrm{d}}{\mathrm{d} (-\alpha\mu/b) \mu} > 0
        \\ &\iff \alpha e^{(\alpha-1)\mu/b} \frac{\mathrm{d}}{\mathrm{d} \mu} ((\alpha-1)\mu/b) + (\alpha-1) e^{-\alpha\mu/b} \frac{\mathrm{d}}{\mathrm{d} \mu} (-\alpha\mu/b) > 0
        \\ &\iff \alpha e^{(\alpha-1)\mu/b} \frac{\alpha-1}{b} - (\alpha-1) e^{-\alpha\mu/b} \frac{\alpha}{b} > 0
        \iff \frac{\alpha (\alpha-1)}{b} e^{(\alpha-1)\mu/b} - \frac{\alpha (\alpha-1)}{b} e^{-\alpha\mu/b} > 0
        \\ &\iff e^{(\alpha-1)\mu/b} - e^{-\alpha\mu/b} > 0
        \iff (\alpha-1)\mu/b > -\alpha\mu/b
        \iff \alpha-1 > -\alpha
        \iff \alpha > \frac{1}{2},
    \end{align*}
    which is always true.
\end{proof}

\begin{theorem}[Biased Laplace mechanism for a single e-value; Theorem~3.2 in the main text]
    For any $\alpha > 1$ and $\epsilon > 0$, let
    \[ E^\DP(D) = E(D) \cdot e^{-\xi}, \qquad \xi \sim \mathrm{Laplace}\bigl( -\log (1 - b_{\alpha,\epsilon}^2), b_{\alpha,\epsilon} \bigr), \]
    where
    \begin{align*}
        b_{\alpha,\epsilon} &\coloneq 1 \,/\, \mathbf{h}_\alpha^{-1}\left( (2\alpha - 1) e^{(\alpha - 1) \epsilon} \right),
        \\
        \mathbf{h}_\alpha(t) &\coloneq \alpha e^{(\alpha-1) \Delta_{\log}(E) t} + (\alpha-1) e^{-\alpha \Delta_{\log}(E) t} \qquad \text{for }t \geq 0.
    \end{align*}
    Then, as long as $b_{\alpha,\epsilon} < 1$,
    $E^\DP(D)$ is a valid e-value satisfying $(\alpha, \epsilon)$-Rényi differential privacy.
\end{theorem}

\begin{proof}
    We are considering $\xi \sim \mathrm{Laplace}(\mu, b)$ for some $\mu$ and $b$.
    By the post-processing theorem,
    $E^\DP(D) = E(D) \cdot e^{-\xi}$ is $(\alpha, \epsilon)$-RDP iff $\log E^\DP(D) = \log \left( E(D) \cdot e^{-\xi} \right) = \log E(D) - \xi$ is $(\alpha, \epsilon)$-RDP; note that $(\log E(D) - xi) \sim \mathrm{Laplace}(\log E(D) - \mu, b)$. It then follows, using using Lemma~\ref{thm:renyi-laplace}, that
    \begin{align*}
        & D_\alpha (\log E^\DP(D) \Mid \log E^\DP(D'))
        \\ &= D_\alpha (\mathrm{Laplace}(\log E(D) - \mu, b) \Mid \mathrm{Laplace}(\log E(D') - \mu, b))
        \\ &= D_\alpha (\mathrm{Laplace}(0, b) \Mid \mathrm{Laplace}((\log E(D') - \mu) - (\log E(D) - \mu), b))
        \\ &= D_\alpha (\mathrm{Laplace}(0, b) \Mid \mathrm{Laplace}(\log E(D') - \log E(D), b))
        \\ &= \frac{1}{\alpha - 1} \log \left(
            \frac{\alpha}{2\alpha - 1} e^{ (\alpha - 1) \lvert \log E(D') - \log E(D) \rvert /b }
            + \frac{\alpha-1}{2\alpha-1} e^{ -\alpha \lvert \log E(D') - \log E(D) \rvert /b }
        \right).
    \end{align*}
    For this to be at most $\epsilon$, we require:
    \begin{align*}
        & \frac{1}{\alpha - 1} \log \left(
            \frac{\alpha}{2\alpha - 1} e^{ (\alpha - 1) \lvert \log E(D') - \log E(D) \rvert /b }
            + \frac{\alpha-1}{2\alpha-1} e^{ -\alpha \lvert \log E(D') - \log E(D) \rvert /b }
        \right) \leq \epsilon
        \\ \impliedby&
            \frac{1}{\alpha - 1} \log \left(
                \frac{\alpha}{2\alpha - 1} e^{ (\alpha - 1) \Delta_{\log}(E) / b }
                + \frac{\alpha-1}{2\alpha-1} e^{ -\alpha \Delta_{\log}(E) /b }
            \right) \leq \epsilon
        \\ \iff&
            \frac{\alpha}{2\alpha - 1} e^{ (\alpha - 1) \Delta_{\log}(E) / b }
            + \frac{\alpha-1}{2\alpha-1} e^{ -\alpha \Delta_{\log}(E) /b }
            \leq e^{(\alpha - 1) \epsilon}
        \\ \iff&
            \alpha e^{ (\alpha - 1) \Delta_{\log}(E) / b }
            + (\alpha-1) e^{ -\alpha \Delta_{\log}(E) /b }
            \leq (2\alpha - 1) e^{(\alpha - 1) \epsilon}.
    \end{align*}
    Let $\mathbf{h}_\alpha(t) = \alpha e^{(\alpha-1) \Delta_{\log}(E) t} + (\alpha-1) e^{-\alpha \Delta_{\log}(E) t}$, for $t \geq 0$.
    Note that this is strictly increasing in $t$, as
    \begin{align*}
        \frac{\mathrm{d}}{\mathrm{d} t} \mathbf{h}_\alpha(t)
        &= \frac{\mathrm{d}}{\mathrm{d} t} \alpha e^{(\alpha-1) c t} + \frac{\mathrm{d}}{\mathrm{d} t} (\alpha-1) e^{-\alpha c t}
        \\ &= \alpha e^{(\alpha-1) c t} (\alpha-1) c + (\alpha-1) e^{-\alpha c t} (-\alpha c)
        \\ &= \alpha (\alpha-1) c \left( e^{(\alpha-1) c t} - e^{-\alpha c t} \right) > 0
        \\ \iff& e^{(\alpha-1) c t} > e^{-\alpha c t}
        \iff (\alpha-1) c t > -\alpha c t
        \\ \iff& \alpha-1 > -\alpha
        \iff \alpha > \frac{1}{2},
    \end{align*}
    which is always true.
    Hence, the inverse $\mathbf{h}^{-1}_\alpha(y)$ exists.
    This allows us to write
    \begin{align*}
        &
            \alpha e^{ (\alpha - 1) \Delta_{\log}(E) / b }
            + (\alpha-1) e^{ -\alpha \Delta_{\log}(E) /b }
            \leq (2\alpha - 1) e^{(\alpha - 1) \epsilon}
        \\ \iff& \mathbf{h}_\alpha(1/b) \leq (2\alpha - 1) e^{(\alpha - 1) \epsilon}
        \\ \iff& 1/b \leq \mathbf{h}_\alpha^{-1}\left( (2\alpha - 1) e^{(\alpha - 1) \epsilon} \right)
        \\ \iff& b \geq 1 / \mathbf{h}_\alpha^{-1}\left( (2\alpha - 1) e^{(\alpha - 1) \epsilon} \right).
    \end{align*}

    So, taking $b = 1 / \mathbf{h}_\alpha^{-1}\left( (2\alpha - 1) e^{(\alpha - 1) \epsilon} \right)$ ensures $(\alpha, \epsilon)$-Renyi differential privacy.

    Thus, all that remains is to find the smallest bias $\mu$ that ensures validity of $E^\DP(D)$.
    As argued in Section~3.1, all we need is that the MGF of $\xi$ is at most one at $-1$.
    The MGF of a Laplace distribution with mean $\mu$ and scale parameter $b$ is given by $t \mapsto e^{t \mu} / (1 - b^2 t^2)$, defined at $\lvert t \rvert < 1/b$.
    Thus we need that $\lvert -1 \rvert = 1 < 1/b$, i.e., that $b < 1$.
    As long as this is satisfied, it then follows:
    \begin{align*}
        \E[e^{-\xi}] &= e^{-\mu} / \left( 1 - b^2 \right) \leq 1
        \\ \iff& e^{-\mu} \leq 1 - b^2
        \\ \iff& \mu \geq -\log (1 - b^2).
        \qedhere
    \end{align*}
\end{proof}

\begin{proposition}[Proposition~3.3 in the main text]
    Let $E^\DP(D)$ be as in Equation~1 in the main text. Then
    \[ \E\left[\frac{1}{n} \log E^\DP(D)\right] = \E\left[\frac{1}{n} \log E(D)\right] - \frac{\E[\xi]}{n}. \]
\end{proposition}
\begin{proof}
    \begin{align*}
        \E\left[\frac{1}{n} \log E^\DP(D)\right]
        = \E\left[\frac{1}{n} \log \bigl( E(D) \cdot e^{-\xi} \bigr)\right]
        &= \E\left[\frac{1}{n} \log E(D) - \frac{1}{n} \xi\right]
        \\ &= \E\left[\frac{1}{n} \log E(D)\right] - \frac{\E[\xi]}{n}. \qedhere
    \end{align*}
\end{proof}

\begin{proposition}[Optional continuation; Proposition~3.4 in the main text]
    If $E^\DP_1(D_1)$ and $E^\DP_2(D_2)$ are $(\alpha, \epsilon)$-Rényi differentially private e-values for a null hypothesis $H_0$, with data $D_1$ independent from $D_2$, then $E^\DP_1(D_1) \cdot E^\DP_2(D_2)$ is also an $(\alpha, \epsilon)$-Rényi differentially private e-value for $H_0$. Moreover, the release of both $E^\DP_1(D_1)$ and $E^\DP_1(D_1) \cdot E^\DP_2(D_2)$ is also $(\alpha, \epsilon)$-Rényi differentially private.
\end{proposition}

\begin{proof}
    Let $(D'_1, D'_2)$ be a neighboring dataset to $(D_1, D_2)$, i.e., either (i) $D'_1$ and $D_1$ differ by a single element and $D'_2 = D_2$; (ii) $D'_2$ and $D_2$ differ by a single element and $D'_1 = D_1$; or (iii) $D'_1 = D_1$ and $D'_2 = D_2$.
    And since $E^\DP_1(\cdot)$ and $E^\DP_2(\cdot)$ are $(\alpha, \epsilon)$-Rényi differentially private, we have that
    \[
        D_\alpha\bigl( E^\DP_1(D'_1) \Mid E^\DP_1(D_1) \bigr) \leq \epsilon
        \quad\text{and}\quad
        D_\alpha\bigl( E^\DP_1(D'_1) \Mid E^\DP_1(D_1) \bigr) \leq \epsilon.
    \]
    Then, using the fact that the randomness in $E^\DP_1(\cdot)$ and $E^\DP_2(\cdot)$ (due to the noise) are independent,
    \begin{align*}
        & D_\alpha\left( (E^\DP_1(D_1), E^\DP_1(D_2)) \Mid (E^\DP_1(D'_1), E^\DP_1(D'_2)) \right)
        \\ &= D_\alpha\left( E^\DP_1(D'_1) \Mid E^\DP_1(D'_1) \right)
            + D_\alpha\left( E^\DP_1(D_2) \Mid E^\DP_1(D'_2) \right);
    \end{align*}
    and since we must be in one of the three cases delineated above, this must be at most $\epsilon$.
    Hence, $(E^\DP_1(\cdot), E^\DP_2(\cdot))$ is $(\alpha, \epsilon)$-Rényi differentially private.

    Now considering randomness over the data, since the two datasets are assumed to be independent it is immediate to see that the product is an e-value. Under the null,
    \[ \E[E^\DP_1(D_1) \cdot E^\DP_2(D_2)] = \E[E^\DP_1(D_1)] \cdot \E[E^\DP_2(D_2)] \leq 1 \cdot 1 = 1. \]
    And, by the post-processing theorem, $(E^\DP_1(\cdot), E^\DP_1(\cdot) \cdot E^\DP_2(\cdot))$ is also $(\alpha, \epsilon)$-Rényi DP.
\end{proof}

\begin{proposition}[E-to-p conversion; Proposition~3.5 in the main text]
    If $E^\DP(D)$ is an $(\alpha, \epsilon)$-Rényi differentially private e-value for a null hypothesis $H_0$, then $1/E^\DP(D)$ is an $(\alpha, \epsilon)$-Rényi differentially private post-hoc valid p-value for $H_0$.
\end{proposition}

\begin{proof}
    By Theorem~2 of \citep{evalue-posthoc}, $1/E^\DP(D)$ is a post-hoc valid p-value. All that remains is to show that it satisfies $(\epsilon, \delta)$-Rényi differential privacy, which follows directly by the post-processing theorem for $t \mapsto 1/t$.
\end{proof}

\begin{proposition}[Independent Averaging; Proposition~3.6 in the main text]
    If $E^\DP_1(D_1)$ and $E^\DP_2(D_2)$ are $(\alpha, \epsilon)$-Rényi differentially private e-values for a null hypothesis $H_0$, then for any $\eta \in [0, 1]$, $\eta E^\DP_1(D_1) + (1-\eta) E^\DP_2(D_2)$ is also an $(\alpha, \epsilon)$-differentially private e-value for $H_0$.
\end{proposition}

\begin{proof}
    Let $(D'_1, D'_2)$ be a neighboring dataset to $(D_1, D_2)$, i.e., either (i) $D'_1$ and $D_1$ differ by a single element and $D'_2 = D_2$; (ii) $D'_2$ and $D_2$ differ by a single element and $D'_1 = D_1$; or (iii) $D'_1 = D_1$ and $D'_2 = D_2$.
    And since $E^\DP_1(\cdot)$ and $E^\DP_2(\cdot)$ are $(\alpha, \epsilon)$-Rényi differentially private, we have that
    \[
        D_\alpha\bigl( E^\DP_1(D'_1) \Mid E^\DP_1(D_1) \bigr) \leq \epsilon
        \quad\text{and}\quad
        D_\alpha\bigl( E^\DP_1(D'_1) \Mid E^\DP_1(D_1) \bigr) \leq \epsilon.
    \]
    Then, using the fact that the randomness in $E^\DP_1(\cdot)$ and $E^\DP_2(\cdot)$ (due to the noise) are independent,
    \begin{align*}
        & D_\alpha\left( (E^\DP_1(D_1), E^\DP_1(D_2)) \Mid (E^\DP_1(D'_1), E^\DP_1(D'_2)) \right)
        \\ &= D_\alpha\left( E^\DP_1(D'_1) \Mid E^\DP_1(D'_1) \right)
            + D_\alpha\left( E^\DP_1(D_2) \Mid E^\DP_1(D'_2) \right);
    \end{align*}
    and since we must be in one of the three cases delineated above, this must be at most $\epsilon$.
    Hence, $(E^\DP_1(\cdot), E^\DP_2(\cdot))$ is $(\alpha, \epsilon)$-Rényi differentially private.

    Now considering randomness over the data, it is immediate to see that the convex combination is an e-value: under the null,
    \[ \E[\eta E^\DP_1(D_1) + (1-\eta) E^\DP_2(D_2)] = \eta \E[E^\DP_1(D_1)] + (1-\eta) \E[E^\DP_2(D_2)] \leq \eta 1 + (1-\eta) 1 = 1. \]
    And, by the post-processing theorem, $(E^\DP_1(\cdot), \eta E^\DP_1(\cdot) + (1-\eta) E^\DP_2(\cdot))$ is also $(\alpha, \epsilon)$-Rényi DP.
\end{proof}

\begin{proposition}[Dependent Averaging; Proposition~3.7 in the main text]
    If $E^\DP_1(D)$ and $E^\DP_2(D)$ are $(\alpha, \epsilon)$-Rényi differentially private e-values for a null hypothesis $H_0$, then for any $\eta \in [0, 1]$, $\eta E^\DP_1(D) + (1-\eta) E^\DP_2(D)$ is an $(\alpha, 2\epsilon)$-differentially private e-value for $H_0$.
\end{proposition}

\begin{proof}
    By the composition theorem, the release of $(E^\DP_1(D), E^\DP_2(D))$ is $(\alpha, 2\epsilon)$-RDP. The post-processing then ensures that the release of $\eta E^\DP_1(D) + (1-\eta) E^\DP_2(D)$ is also $(\alpha, 2\epsilon)$-RDP.

    The validity of the e-value is due to linearity: under the null,
    \[ \E[\eta E^\DP_1(D_1) + (1-\eta) E^\DP_2(D_2)] = \eta \E[E^\DP_1(D_1)] + (1-\eta) \E[E^\DP_2(D_2)] \leq \eta 1 + (1-\eta) 1 = 1. \]
\end{proof}

\begin{theorem}[Theorem~3.9 in the main text]
    Under Assumption~3.8,
    let $\alpha > 1$ and $\epsilon > 0$.
    For each $j = 1, \ldots, k$, let $E^\DP_j$ be an $(\alpha, \epsilon/k)$-Rényi differentially private e-value for the null $H_0^{(j)}$ (e.g., obtained through Theorems~3.1 and 3.2).
    Then $\algo(E^\DP_1, \ldots, E^\DP_k)$ is valid and $(\alpha, \epsilon)$-Rényi differentially private.
\end{theorem}

\begin{proof}
    That $\algo(E^\DP_1, \ldots, E^\DP_k)$ is valid is immediate: by Assumption~3.8 it suffices to have $E^\DP_1, \ldots, E^\DP_k$ be valid, which we assume them to be.

    As for privacy, this follows as an immediate use of the composition and post-processing theorems (Propositions~2.3 and 2.4): since $E^\DP_1, \ldots, E^\DP_k$ are each $(\alpha, \epsilon/k)$-Rényi differentially private, the tuple $(E^\DP_1, \ldots, E^\DP_k)$ is $(\alpha, \epsilon)$-Rényi differentially private. Thus, by post-processing, $\algo(E^\DP_1, \ldots, E^\DP_k)$ is also $(\alpha, \epsilon)$-Rényi differentially private.
\end{proof}

\begin{corollary}[Corollary~3.10 in the main text]
    Suppose $\log E_\theta(D)$ is locally Lipschitz in $\theta$.
    Let $\alpha > 1$ and $\epsilon > 0$,
    and let $\widetilde{E}^\DP_{\theta_1}, \ldots, \widetilde{E}^\DP_{\theta_k}$ be $(\alpha, \epsilon/k)$-Rényi differentially private versions of $\widetilde{E}_{\theta_1}, \ldots, \widetilde{E}_{\theta_k}$ as defined in Equation~3 in the main text.
    Then $\CI_\alpha\bigl( \widetilde{E}^\DP_{\theta_1}(D), \ldots, \widetilde{E}^\DP_{\theta_k}(D) \bigr)$ is an $(\alpha, \epsilon)$-Rényi differentially private confidence interval for $\theta^\star$, i.e., it satisfies $(\alpha, \epsilon)$-Rényi differential privacy, and
    \[ \P\left[ \theta^\star \in \CI_\alpha\!\left(\widetilde{E}^\DP_{\theta_1}(D), \ldots, \widetilde{E}^\DP_{\theta_k}(D)\right) \right] \geq 1 - \alpha. \]
\end{corollary}

\begin{proof}
    As noted in the main text, each $\widetilde{E}_{\theta_j}(D)$ is an simultaneously an e-value for all nulls $H_0 : \theta^\star = \theta'$, $\theta' \in [a_{j-1}, a_j]$.
    Thus, the $\widetilde{CI}_\alpha(\cdot, \ldots, \cdot)$ algorithm constructs valid confidence intervals:
    \begin{align*}
        \P\left[ \theta^\star \in \CI_\alpha\!\left(\widetilde{E}_{\theta_1}(D), \ldots, \widetilde{E}_{\theta_k}(D)\right) \right]
        &= \P\Biggl[ \theta^\star \in \bigcup_{\substack{j=1 \\ \widetilde{E}_{\theta_j}(D) \leq 1/\alpha }}^k [a_{j-1}, a_j] \Biggr]
        \\ &\geq \P\left[ \widetilde{E}_{\theta_{j'}}(D) \leq 1/\alpha \right], \quad \text{for } j' \text{ s.t. } \theta^\star \in [a_{j'-1}, a_{j'}]
        \\ &= 1 - \P\left[ \widetilde{E}_{\theta_{j'}}(D) > 1/\alpha \right]
        \\ &\geq 1 - \frac{\E[\widetilde{E}_{\theta_{j'}}(D)]}{1/\alpha} \geq 1 - \frac{1}{1/\alpha} = 1 - \alpha.
    \end{align*}
    By applying Theorem~3.9, we conclude.
\end{proof}

\begin{proposition}[Proposition~4.1 in the main text]
    Let $F$ have support in $[\lambda_{\inf}, \lambda_{\sup}] \subset (-1/(1-\theta), 1/\theta)$.
    The log-sensitivity of $E_\theta(D)$, $\Delta_{\log}(E_{\theta})$, is upper bounded as
    \begin{align*}
        \Delta_{\log}(E_{\theta}) &\leq
            \max\bigl\{
            \log ( 1 + \max\{\lambda_{\sup} (1-\theta), -\lambda_{\inf} \theta\} ),
            \\ &\qquad\qquad -\log ( 1 + \min\{\lambda_{\inf}(1-\theta),-\lambda_{\sup}\theta\} )
            \bigr\}.
    \end{align*}
\end{proposition}
\begin{proof}
    Let $D$ and $D'$ be neighboring datasets. I.e., either $D'$ is $D$ with an additional sample $Y_{n+1}$, or it is $D$ with a sample $Y_i$ removed.
    Since the e-value $E_{D,\theta}$ is invariant to permutations of the data, we can, without loss of generality, consider that $Y_n$ was the element removed.

    When adding one new element, it follows:
    \begin{align*}
        \lvert \log E_{D,\theta} - \log E_{D',\theta} \rvert
        &= \left\lvert \sum_{i=1}^n \log \bigl( 1 + \lambda_i (Y_i - \theta) \bigr) - \sum_{i=1}^{n+1} \log \bigl( 1 + \lambda_i (Y_i - \theta) \bigr) \right\rvert
        \\ &= \left\lvert \log \bigl( 1 + \lambda_{n+1} (Y_{n+1} - \theta) \bigr) \right\rvert;
    \end{align*}
    and, when removing an element -- wlog. $Y_n$, we have
    \begin{align*}
        \lvert \log E_{D,\theta} - \log E_{D',\theta} \rvert
        &= \left\lvert \sum_{i=1}^n \log \bigl( 1 + \lambda_i (Y_i - \theta) \bigr) - \sum_{i=1}^{n-1} \log \bigl( 1 + \lambda_i (Y_i - \theta) \bigr) \right\rvert
        \\ &= \left\lvert \log \bigl( 1 + \lambda_{n} (Y_{n} - \theta) \bigr) \right\rvert.
    \end{align*}
    Hence, all that remains is to upper bound $\left\lvert \log \bigl( 1 + \lambda (Y - \theta) \bigr) \right\rvert$, for any $\lambda \in (-c/(1-\theta), c/\theta)$ and $Y \in [0, 1]$.

    We split into a few cases:
    \begin{enumerate}
        \item When $Y \geq \theta$ and $\lambda \geq 0$, it follows:
            \begin{align*}
                \log \bigl( 1 + \lambda (Y - \theta) \bigr)
                &\leq \log \bigl( 1 + \lambda (1 - \theta) \bigr)
                \leq \log \bigl( 1 + \lambda_{\sup} (1 - \theta) \bigr)
                \\
                \log \bigl( 1 + \lambda (Y - \theta) \bigr)
                &\geq \log \bigl( 1 + \lambda 0 \bigr)
                = \log 1 = 0;
            \end{align*}
        \item When $Y \geq \theta$ and $\lambda \leq 0$, it follows:
            \begin{align*}
                \log \bigl( 1 + \lambda (Y - \theta) \bigr)
                &\leq \log \bigl( 1 + \lambda 0 \bigr)
                = \log 1 = 0
                \\
                \log \bigl( 1 + \lambda (Y - \theta) \bigr)
                &\geq \log \bigl( 1 + \lambda (1 - \theta) \bigr)
                \geq \log \bigl( 1 + \lambda_{\inf} (1 - \theta) \bigr)
            \end{align*}
        \item When $Y \leq \theta$ and $\lambda \geq 0$, it follows:
            \begin{align*}
                \log \bigl( 1 + \lambda (Y - \theta) \bigr)
                &\leq \log \bigl( 1 + \lambda 0 \bigr)
                = \log 1 = 0
                \\
                \log \bigl( 1 + \lambda (Y - \theta) \bigr)
                &\geq \log \bigl( 1 + \lambda (0 - \theta) \bigr)
                \geq \log \bigl( 1 + \lambda_{\sup} (0 - \theta) \bigr)
                = \log \bigl( 1 - \lambda_{\sup} \theta \bigr)
            \end{align*}
        \item When $Y \leq \theta$ and $\lambda \leq 0$, it follows:
            \begin{align*}
                \log \bigl( 1 + \lambda (Y - \theta) \bigr)
                &\leq \log \bigl( 1 + \lambda (0 - \theta) \bigr)
                \leq \log \bigl( 1 + \lambda_{\inf} (0 - \theta) \bigr)
                \leq \log \bigl( 1 - \lambda_{\inf} \theta \bigr)
                \\
                \log \bigl( 1 + \lambda (Y - \theta) \bigr)
                &\geq \log \bigl( 1 + \lambda 0 \bigr)
                = \log 1 = 0;
            \end{align*}
    \end{enumerate}
    Thus, we conclude that
    \begin{align*}
        \left\lvert \log \bigl( 1 + \lambda (Y - \theta) \bigr) \right\rvert
        &\leq \max \{ 0, \lvert \log(1 + \lambda_{\sup} (1-\theta)) \rvert, \lvert \log(1 + \lambda_{\inf} (1-\theta)) \rvert,
            \\ &\qquad\quad \lvert \log(1 - \lambda_{\sup} \theta) \rvert, \lvert \log(1 - \lambda_{\inf} \theta) \rvert \}
        \\ &= \max \{ \lvert \log(1 + \lambda_{\sup} (1-\theta)) \rvert, \lvert \log(1 + \lambda_{\inf} (1-\theta)) \rvert,
            \\ &\qquad\quad \lvert \log(1 - \lambda_{\sup} \theta) \rvert, \lvert \log(1 - \lambda_{\inf} \theta) \rvert \}.
        \\ &= \max\bigl\{
            \log ( 1 + \max \{ \lambda_{\sup} (1-\theta), \lambda_{\inf} (1-\theta), -\lambda_{\sup} \theta, -\lambda_{\inf} \theta \} ),
            \\ &\qquad\quad -\log ( 1 + \min \{ \lambda_{\sup} (1-\theta), \lambda_{\inf} (1-\theta), -\lambda_{\sup} \theta, -\lambda_{\inf} \theta \} )
            \bigr\}.
        \\ &\leq \max\bigl\{
            \log ( 1 + \max\{\lambda_{\sup} (1-\theta), -\lambda_{\inf} \theta\} ),
            \\ &\qquad\quad -\log ( 1 + \min\{\lambda_{\inf}(1-\theta),-\lambda_{\sup}\theta\} )
            \bigr\}.
        \qedhere
    \end{align*}
\end{proof}

\begin{proposition}[Proposition~4.2 in the main text]
    For distributions $F$ with support contained in $[\lambda_{\inf}, \lambda_{\sup}]$, the Lipschitz constant of $\log E_\theta(D)$ over $\theta \in [\theta_{\inf}, \theta_{\sup}]$ is upper bounded by
    \[
        \|\theta \mapsto \log E_\theta(D)\|_{\mathrm{Lip}}
        \leq \max\left\{
            \left\lvert \frac{\lambda_{\sup}}{1 - \lambda_{\sup} \theta_{\sup}} \right\rvert,
            \left\lvert \frac{\lambda_{\inf}}{1 + \lambda_{\inf} (1 - \theta_{\inf})} \right\rvert
            \right\}.
    \]
\end{proposition}

\begin{proof}
    Bounding the Lipschitz constant is equivalent to bounding the absolute value of the derivative. So it follows:
    \begin{align*}
        \left\lvert \frac{\d}{\d\theta} \log E_\theta(D) \right\rvert
        &= \left\lvert \frac{\frac{\d}{\d\theta} E_\theta(D)}{E_\theta(D)} \right\rvert
        = \left\lvert \frac{\frac{\d}{\d\theta} \E_{\lambda \sim F}[ \prod_{i=1}^n (1 + \lambda (Y_i - \theta)) ]}{E_\theta(D)} \right\rvert
        \\ &= \left\lvert \frac{\E_{\lambda \sim F}[ \frac{\d}{\d\theta} \prod_{i=1}^n (1 + \lambda (Y_i - \theta)) ]}{E_\theta(D)} \right\rvert
        \\ &= \left\lvert \frac{\E_{\lambda \sim F}[ \frac{\d}{\d\theta} \exp( \sum_{i=1}^n \log (1 + \lambda (Y_i - \theta)) ) ]}{E_\theta(D)} \right\rvert
        \\ &= \left\lvert \frac{\E_{\lambda \sim F}[ \exp( \sum_{i=1}^n \log (1 + \lambda (Y_i - \theta)) ) \frac{\d}{\d\theta} \sum_{i=1}^n \log (1 + \lambda (Y_i - \theta)) ]}{E_\theta(D)} \right\rvert
    \end{align*}
    Write $E^{(\lambda)}_\theta(D) = \prod_{i=1}^n (1 + \lambda (Y_i - \theta))$. Then:
    \begin{align*}
        & \left\lvert \frac{\E_{\lambda \sim F}[ \exp( \sum_{i=1}^n \log (1 + \lambda (Y_i - \theta)) ) \frac{\d}{\d\theta} \sum_{i=1}^n \log (1 + \lambda (Y_i - \theta)) ]}{E_\theta(D)} \right\rvert
        \\ &= \left\lvert \frac{\E_{\lambda \sim F}[ E^{(\lambda)}_\theta(D) \frac{\d}{\d\theta} \sum_{i=1}^n \log (1 + \lambda (Y_i - \theta)) ]}{E_\theta(D)} \right\rvert
        \\ &= \left\lvert \frac{\E_{\lambda \sim F}[ E^{(\lambda)}_\theta(D) \frac{\d}{\d\theta} \sum_{i=1}^n \log (1 + \lambda (Y_i - \theta)) ]}{\E_{\lambda \sim F}[E^{(\lambda)}_\theta(D)]} \right\rvert;
    \end{align*}
    Now let $F_{n,\theta}$ be the distribution such that $\d F_{n,\theta} / \d F \propto E^{(\lambda)}_\theta(D)$. This allows us to rewrite the ratio as
    \begin{align*}
        & \left\lvert \frac{\E_{\lambda \sim F}[ E^{(\lambda)}_\theta(D) \frac{\d}{\d\theta} \sum_{i=1}^n \log (1 + \lambda (Y_i - \theta)) ]}{\E_{\lambda \sim F}[E^{(\lambda)}_\theta(D)]} \right\rvert
        \\ &= \left\lvert \E_{\lambda \sim F_{n,\theta}}\left[ \frac{\d}{\d\theta} \sum_{i=1}^n \log (1 + \lambda (Y_i - \theta)) \right] \right\rvert
        \\ &= \left\lvert \E_{\lambda \sim F_{n,\theta}}\left[ \sum_{i=1}^n \frac{-\lambda}{1 + \lambda (Y_i - \theta)} \right] \right\rvert.
    \end{align*}
    So all that remains is to bound this.
    Now, note that $-\lambda / (1 + \lambda (Y_i - \theta))$ is (i) increasing in $Y - \theta$ and (ii) decreasing in $\lambda$. Thus we have
    \[
        -\frac{\lambda_{\sup}}{1 - \lambda_{\sup} \theta_{\sup}}
        \leq -\lambda / (1 + \lambda (Y_i - \theta))
        \leq -\frac{\lambda_{\inf}}{1 + \lambda_{\inf} (1 - \theta_{\inf})},
    \]
    and
    \begin{align*}
        \left\lvert \E_{\lambda \sim F_{n,\theta}}\left[ \sum_{i=1}^n \frac{-\lambda}{1 + \lambda (Y_i - \theta)} \right] \right\rvert
        &\leq \max\left\{
            \left\lvert -\frac{\lambda_{\sup}}{1 - \lambda_{\sup} \theta_{\sup}} \right\rvert,
            \left\lvert -\frac{\lambda_{\inf}}{1 + \lambda_{\inf} (1 - \theta_{\inf})} \right\rvert
            \right\}
        \\ &= \max\left\{
            \left\lvert \frac{\lambda_{\sup}}{1 - \lambda_{\sup} \theta_{\sup}} \right\rvert,
            \left\lvert \frac{\lambda_{\inf}}{1 + \lambda_{\inf} (1 - \theta_{\inf})} \right\rvert
            \right\}.
    \end{align*}
\end{proof}

\begin{proposition}[Proposition~4.3 in the main text]
    Suppose the scores are all contained in $[a, b]$ with $a > 0$.
    Then the log-sensitivity of $E^\exch(D; S^\test)$, $\Delta_{\log}(E^\exch)$, is upper bounded as
    \[ \Delta_{\log}(E^\exch) \leq 2 \cdot \frac{b/a}{n+1}. \]
\end{proposition}
\begin{proof}
    We want to compute $\Delta_{\log}(E^\exch) = \sup_{|D \Delta D'| \leq 1} \lvert \log E^\exch(D; S^\test) - \log E^\exch(D'; S^\test) \rvert$.
    For simplicity, we will rewrite the e-value as
    \[ E^\exch(D; S^\test) = \frac{S^\test}{\frac{1}{n+1}\left( \sum_{i=1}^n S_i + S^\test \right)} \]
    There are three cases:
    \begin{enumerate}
        \item[(i)] $D = D'$, in which case $\lvert \log E^\exch(D; S^\test) - \log E^\exch(D'; S^\test) \rvert = 0$.

        \item[(ii)] $D'$ differs from $D$ by one new element, $Y_{n+1}$; in this case we have
            \begin{align*}
                & \lvert \log E^\exch(D; S^\test) - \log E^\exch(D'; S^\test) \rvert
                \\ &= \left\lvert \left( \log S^\test - \log \frac{ \sum_{i=1}^n S_i + S^\test }{n+1} \right) - \left( \log S^\test - \log \frac{\sum_{i=1}^{n+1} S_i + S^\test}{n+2} \right) \right\rvert
                \\ &= \left\lvert \log \frac{ \sum_{i=1}^n S_i + S^\test }{n+1} - \log \frac{\sum_{i=1}^{n+1} S_i + S^\test}{n+2} \right\rvert
                \\ &\leq \left\lvert \frac{ \sum_{i=1}^n S_i + S^\test }{n+1} - \frac{\sum_{i=1}^{n+1} S_i + S^\test}{n+2} \right\rvert \biggm/ \left( \frac{\sum_{i=1}^n a + a}{n+2} \right)
                \\ &= \left\lvert \frac{ \sum_{i=1}^n S_i + S^\test }{n+1} - \frac{\sum_{i=1}^{n+1} S_i + S^\test}{n+2} \right\rvert \biggm/ \left( \frac{n+1}{n+2} a \right)
                \\ &= \left\lvert \frac{ \sum_{i=1}^n S_i + S^\test }{n+1} \cdot \frac{n+2}{n+1} - \frac{\sum_{i=1}^{n+1} S_i + S^\test}{n+2} \cdot \frac{n+2}{n+1} \right\rvert / a
                \\ &= \left\lvert \frac{ \sum_{i=1}^n S_i + S^\test }{n+1} \cdot \frac{n+2}{n+1} - \frac{\sum_{i=1}^{n+1} S_i + S^\test}{n+1} \right\rvert / a
                \\ &= \left\lvert \sum_{i=1}^n \left( \frac{1}{n+1} \cdot \frac{n+2}{n+1} - \frac{1}{n+1} \right) S_i + \left( \frac{1}{n+1} \cdot \frac{n+2}{n+1} - \frac{1}{n+1} \right) S^\test - \frac{S_{n+1}}{n+1} \right\rvert / a
                \\ &= \left\lvert \sum_{i=1}^n \left( \frac{n+2}{n+1} - 1 \right) S_i + \left( \frac{n+2}{n+1} - 1 \right) S^\test - S_{n+1} \right\rvert / a (n+1)
            \end{align*}
            \begin{align*}
                \\ &\leq \left( \left( \frac{n+2}{n+1} - 1 \right) \left\lvert \sum_{i=1}^n S_i + S^\test \right\rvert + \lvert S_{n+1} \rvert \right) / a (n+1)
                \\ &= \left( \frac{1}{n+1} \left\lvert \sum_{i=1}^n S_i + S^\test \right\rvert + \lvert S_{n+1} \rvert \right) / a (n+1)
                \\ &\leq \left( \frac{1}{n+1} \left\lvert \sum_{i=1}^n b + b \right\rvert + b \right) / a (n+1)
                \\ &= \left( b + b \right) / a (n+1) = 2b / a(n+1) = 2 \cdot \frac{b/a}{n+1}.
            \end{align*}

        \item[(iii)] $D'$ differs from $D$ by one less element; in this case we have
            \begin{align*}
                & \lvert \log E^\exch(D; S^\test) - \log E^\exch(D'; S^\test) \rvert
                \\ &= \left\lvert \left( \log S^\test - \log \frac{ \sum_{i=1}^n S_i + S^\test }{n+1} \right) - \left( \log S^\test - \log \frac{\sum_{i=1}^{n-1} S_i + S^\test}{n} \right) \right\rvert
                \\ &= \left\lvert \log \frac{ \sum_{i=1}^n S_i + S^\test }{n+1} - \log \frac{\sum_{i=1}^{n-1} S_i + S^\test}{n} \right\rvert
                \\ &\leq \left\lvert \frac{ \sum_{i=1}^n S_i + S^\test }{n+1} - \frac{\sum_{i=1}^{n-1} S_i + S^\test}{n} \right\rvert \biggm/ \left( \frac{\sum_{i=1}^{n-1} a + a}{n+1} \right)
                \\ &= \left\lvert \frac{ \sum_{i=1}^n S_i + S^\test }{n+1} - \frac{\sum_{i=1}^{n-1} S_i + S^\test}{n} \right\rvert \biggm/ \left( \frac{n}{n+1} a \right)
                \\ &= \left\lvert \frac{ \sum_{i=1}^n S_i + S^\test }{n+1} \cdot \frac{n}{n+1} - \frac{\sum_{i=1}^{n-1} S_i + S^\test}{n} \cdot \frac{n}{n+1} \right\rvert / a
                \\ &= \left\lvert \frac{ \sum_{i=1}^n S_i + S^\test }{n+1} \cdot \frac{n}{n+1} - \frac{\sum_{i=1}^{n-1} S_i + S^\test}{n+1} \right\rvert / a
                \\ &= \left\lvert \sum_{i=1}^{n-1} \left( \frac{1}{n+1} \cdot \frac{n}{n+1} - \frac{1}{n+1} \right) S_i + \left( \frac{1}{n+1} \cdot \frac{n}{n+1} - \frac{1}{n+1} \right) S^\test + \frac{S_n}{n+1} \right\rvert / a
                \\ &= \left\lvert \sum_{i=1}^{n-1} \left( \frac{n}{n+1} - 1 \right) S_i + \left( \frac{n}{n+1} - 1 \right) S^\test + S_n \right\rvert / a (n+1)
                \\ &\leq \left( \left( 1 - \frac{n}{n+1} \right) \left\lvert \sum_{i=1}^{n-1} S_i + S^\test \right\rvert + \lvert S_n \rvert \right) / a (n+1)
                \\ &= \left( \frac{1}{n+1} \left\lvert \sum_{i=1}^{n-1} S_i + S^\test \right\rvert + \lvert S_n \rvert \right) / a (n+1)
                \\ &\leq \left( \frac{1}{n+1} \left\lvert \sum_{i=1}^{n-1} b + b \right\rvert + b \right) / a (n+1)
                \\ &= \left( \frac{n}{n+1} b + b \right) / a (n+1)
                \\ &\leq 2b / a (n+1)
                = 2 \cdot \frac{b/a}{n+1}.
            \end{align*}
    \end{enumerate}
\end{proof}

\section{Additional theoretical results}

\subsection{$(\epsilon, \delta)$-differential privacy} \label{suppl:epsilon-delta-dp}

\begin{theorem}[$(\epsilon, \delta)$-DP Biased Gaussian mechanism for a single e-value] \label{thm:classic-dp-gaussian}
    For any $\epsilon, \delta > 0$,
    let $E^\DP(D)$ be as in Equation~1 in the main text with $\xi \sim \mathcal{N}(c^2 [\Delta_{\log}(E)]^2 / 2 \epsilon^2, c^2 [\Delta_{\log}(E)]^2 / \epsilon^2)$, for $c^2 = 2 \log 1.25/\delta$.
    Then $E^\DP(D)$ is a valid e-value satisfying $(\epsilon, \delta)$-differential privacy.
\end{theorem}

\begin{proof}
    First note that for any bias $\mu$, $E^\DP(D)$ with $\xi \sim \mathcal{N}(\mu, c^2 [\Delta_{\log}(E)]^2 / \epsilon^2)$ is $(\epsilon, \delta)$-differentially private:
    by the post-processing theorem,
    $E^\DP(D) = E(D) \cdot e^{-\xi}$ is $(\epsilon, \delta)$-DP iff $\log E^\DP(D) = \log \left( E(D) \cdot e^{-\xi} \right) = \log E(D) - \xi$ is $(\epsilon, \delta)$-DP iff $\log E(D) - (\xi - \mu)$ is $(\epsilon, \delta)$-DP; and, by Theorem A.1 of \citep{dp-book}, the latter is $(\epsilon, \delta)$-DP.

    Thus, all that remains is to find the smallest bias $\mu$ that ensures validity of $E^\DP(D)$.
    As argued in Section~3.1, all we need is that the MGF of $\xi$ is at most one at $-1$. So it follows, using the form of the MGF of an univariate normal:
    \begin{align*}
        \E[e^{-\xi}] &= \exp\left( -\mu + \frac{c^2 [\Delta_{\log}(E)]^2 / \epsilon^2}{2} \right) \leq 1
        \\ \iff& -\mu + \frac{c^2 [\Delta_{\log}(E)]^2 / \epsilon^2}{2} \leq 0
        \\ \iff& \mu \geq \frac{c^2 [\Delta_{\log}(E)]^2 / \epsilon^2}{2} = \frac{c^2 [\Delta_{\log}(E)]^2}{2\epsilon^2}.
        \qedhere
    \end{align*}
\end{proof}

\begin{theorem}[$(\epsilon, 0)$-DP Biased Laplace mechanism for a single e-value]\label{thm:classic-dp-laplace}
    For any $\epsilon > 0$ satisfying $\epsilon > \Delta_{\log} (E)$,
    let $E^\DP(D)$ be as in Equation~1 in the main text with $\xi \sim \mathrm{Laplace}(-\log( 1 - [\Delta_{\log}(E)]^2/\epsilon^2 ), \Delta_{\log}(E)/\epsilon)$.
    Then $E^\DP(D)$ is a valid e-value satisfying $(\epsilon, 0)$-differential privacy.
\end{theorem}

\begin{proof}
    First note that for any bias $\mu$, $E^\DP(D)$ with $\xi \sim \mathrm{Laplace}(\mu, \Delta_{\log}(E)/\epsilon)$ is $(\epsilon, 0)$-differentially private:
    by the post-processing theorem,
    $E^\DP(D) = E(D) \cdot e^{-\xi}$ is $(\epsilon, 0)$-DP iff $\log E^\DP(D) = \log \left( E(D) \cdot e^{-\xi} \right) = \log E(D) - \xi$ is $(\epsilon, 0)$-DP iff $\log E(D) - (\xi - \mu)$ is $(\epsilon, 0)$-DP; and, by Theorem 3.6 of \citep{dp-book}, the latter is $(\epsilon, 0)$-DP.

    Thus, all that remains is to find the smallest bias $\mu$ that ensures validity of $E^\DP(D)$.
    As argued in Section~3.1, all we need is that the MGF of $\xi$ is at most one at $-1$.
    The MGF of a Laplace distribution with mean $\mu$ and scale parameter $b$ is given by $t \mapsto e^{t \mu} / (1 - b^2 t^2)$, defined at $\lvert t \rvert < 1/b$.
    Thus we need that
    \[ \lvert -1 \rvert = 1 < \frac{1}{ \Delta_{\log}(E) / \epsilon } = \frac{\epsilon}{\Delta_{\log}(E)} \iff \Delta_{\log}(E) < \epsilon. \]
    As long as this is satisfied, it then follows:
    \begin{align*}
        \E[e^{-\xi}] &= e^{-\mu} / (1 - [\Delta_{\log}(E)]^2/\epsilon^2) \leq 1
        \\ \iff& e^{-\mu} \leq 1 - [\Delta_{\log}(E)]^2/\epsilon^2
        \\ \iff& -\mu \leq \log ( 1 - [\Delta_{\log}(E)]^2/\epsilon^2 )
        \\ \iff& \mu \geq -\log ( 1 - [\Delta_{\log}(E)]^2/\epsilon^2 ).
        \qedhere
    \end{align*}
\end{proof}

\subsection{Bias is necessary}

The following is a simple result showing that a biased mechanism is necessary for differentially-private e-values.
We consider here that our non-private e-value has expectation exactly equal to 1 and that $E(D)$ is not already differentially private itself.

\begin{proposition}[Bias is necessary]
    Let $\E[E(D)] = 1$ and assume that $E(D)$ is not already differentially private (either $(\epsilon, \delta)$-DP or $(\alpha, \epsilon)$-DP). Then for $E^\DP(D)$ (as per Equation~1 in the main text) to be a differentially private e-value, we must take some $\xi$ with $\E[\xi] > 0$.
\end{proposition}

\begin{proof}
    Since $E(D)$ is not already differentially private, for $E^\DP(D)$ to be differentially private $\xi$ must not be constant. Then, for $E^\DP(D)$ to be a valid e-value, we require that, under the null,
    \[ \E[E^\DP(D)] = \E[E(D) \cdot e^{-\xi}] = \E[E(D)] \cdot \E[e^{-\xi}] = \E[e^{-\xi}] \leq 1; \]
    and, since $s \mapsto e^{-s}$ is strictly convex, using the strict variant of Jensen's inequality (which holds for strictly convex functions and non-constant random variables),
    \[ 1 \geq \E[e^{-\xi}] > e^{-\E[\xi]}. \]
    Taking the log on both sides, we conclude that
    \[ \E[\xi] > 0. \qedhere \]
\end{proof}

\section{Additional experimental results}

\subsection{Comparison with the test-of-tests}

In this section we compare our framework to the private hypothesis testing method of \citep{tot} (the ``test-of-tests'').
Their work takes as input a (non-private) hypothesis test and privatizes it, guaranteeing $(\epsilon, 0)$-differential privacy.

We consider here hypothesis tests for nulls $H^{(p)}_0 : p^\star = p$, for data $Y_1, \ldots, Y_n \sim \mathrm{Bern}(p^\star)$ with $p^\star = 0.5$.
For this, we have two base hypothesis tests: (i) the e-value test for the mean \citep{evalue-mean}, from which we can get a p-value by taking the reciprocal; and (ii) the exact p-value given by a binomial test.
We compute both of these non-private p-values. As for private procedures, we compute:
\begin{itemize}
    \item The test-of-tests \citep{tot} atop the reciprocal of the e-value for the mean;
    \item The test-of-tests \citep{tot} atop the exact p-value;
    \item Our method for $(\epsilon, 0)$-differential privacy with Laplace noise (Theorem~\ref{thm:classic-dp-laplace}); and
    \item Our method for $(\epsilon, \delta)$-differential privacy with Gaussian noise (Theorem~\ref{thm:classic-dp-gaussian}), for a fixed $\delta = 0.01$.
\end{itemize}

The results can be seen in Tables 1-4, for varying values of $\epsilon$ and of $c$ (the latter being a hyperparameter in the e-value for the mean).
Throughout, we note the better performance of our method, achieving much lower p-values even when compared to the test-of-tests atop the exact test.
However, it should be noted that for our biased Laplace mechanism to be defined (which is required in order to obtain the $(\epsilon, 0)$-differential privacy analogous to the test-of-tests), we need a sufficiently low $c$.

\begin{table}[h!]
    \caption{$\epsilon = 0.1$, $c = 0.2$}
    \footnotesize
    \begin{tabular}{c|cc|cccc}
    \bm{$p$}
    & \textbf{\begin{tabular}{@{}c@{}}1/e-value\\(non-private)\end{tabular}}
    & \textbf{\begin{tabular}{@{}c@{}}Exact p-value\\(non-private)\end{tabular}}
    & \textbf{\begin{tabular}{@{}c@{}}Ours, Laplace\\\end{tabular}}
    & \textbf{\begin{tabular}{@{}c@{}}Ours, Gauss.\\(\bm{$\delta=0.01$})\end{tabular}}
    & \textbf{\begin{tabular}{@{}c@{}}ToT atop\\1/e-value\end{tabular}}
    & \textbf{\begin{tabular}{@{}c@{}}ToT atop\\exact p-value\end{tabular}} \\
    \hline
    0.05 & $4.52 \cdot 10^{-68}$ & $5.81 \cdot 10^{-209}$
    & N/A & \bm{$6.89 \cdot 10^{-60}$} & $1.70 \cdot 10^{-02}$ & $2.22 \cdot 10^{-02}$ \\
    0.10 & $2.47 \cdot 10^{-60}$ & $1.34 \cdot 10^{-161}$
    & N/A & \bm{$3.93 \cdot 10^{-47}$} & $1.74 \cdot 10^{-01}$ & $3.83 \cdot 10^{-01}$ \\
    0.15 & $1.92 \cdot 10^{-49}$ & $1.90 \cdot 10^{-110}$
    & N/A & \bm{$5.44 \cdot 10^{-36}$} & $1.50 \cdot 10^{-01}$ & $2.37 \cdot 10^{-02}$ \\
    0.20 & $7.00 \cdot 10^{-45}$ & $6.20 \cdot 10^{-93}$
    & N/A & \bm{$4.68 \cdot 10^{-33}$} & $6.46 \cdot 10^{-01}$ & $6.66 \cdot 10^{-01}$ \\
    0.25 & $3.45 \cdot 10^{-30}$ & $1.23 \cdot 10^{-48}$
    & N/A & \bm{$1.17 \cdot 10^{-22}$} & $6.72 \cdot 10^{-01}$ & $9.06 \cdot 10^{-01}$ \\
    0.30 & $1.56 \cdot 10^{-26}$ & $4.10 \cdot 10^{-40}$
    & N/A & \bm{$2.72 \cdot 10^{-21}$} & $1.14 \cdot 10^{-02}$ & $5.56 \cdot 10^{-01}$ \\
    0.35 & $4.46 \cdot 10^{-18}$ & $9.59 \cdot 10^{-24}$
    & N/A & \bm{$2.99 \cdot 10^{-09}$} & $9.74 \cdot 10^{-01}$ & $2.76 \cdot 10^{-01}$ \\
    0.40 & $2.17 \cdot 10^{-06}$ & $5.95 \cdot 10^{-08}$
    & N/A & $1.00 \cdot 10^{+00}$ & \bm{$2.38 \cdot 10^{-01}$} & $6.05 \cdot 10^{-01}$ \\
    0.45 & $2.80 \cdot 10^{-02}$ & $1.39 \cdot 10^{-03}$
    & N/A & $1.00 \cdot 10^{+00}$ & $7.63 \cdot 10^{-01}$ & \bm{$7.61 \cdot 10^{-01}$} \\
    0.50 & $1.00 \cdot 10^{+00}$ & $4.67 \cdot 10^{-01}$
    & N/A & $1.00 \cdot 10^{+00}$ & \bm{$4.24 \cdot 10^{-01}$} & $8.59 \cdot 10^{-01}$ \\
    0.55 & $1.00 \cdot 10^{+00}$ & $9.37 \cdot 10^{-02}$
    & N/A & $1.00 \cdot 10^{+00}$ & $9.33 \cdot 10^{-01}$ & \bm{$7.30 \cdot 10^{-01}$} \\
    0.60 & $8.65 \cdot 10^{-09}$ & $1.18 \cdot 10^{-10}$
    & N/A & $1.00 \cdot 10^{+00}$ & $5.91 \cdot 10^{-01}$ & \bm{$2.96 \cdot 10^{-02}$} \\
    0.65 & $1.08 \cdot 10^{-15}$ & $6.29 \cdot 10^{-20}$
    & N/A & \bm{$7.70 \cdot 10^{-05}$} & $7.09 \cdot 10^{-01}$ & $2.24 \cdot 10^{-01}$ \\
    0.70 & $3.46 \cdot 10^{-26}$ & $2.38 \cdot 10^{-39}$
    & N/A & \bm{$1.75 \cdot 10^{-10}$} & $3.93 \cdot 10^{-01}$ & $2.90 \cdot 10^{-01}$ \\
    0.75 & $5.00 \cdot 10^{-34}$ & $1.35 \cdot 10^{-58}$
    & N/A & \bm{$1.69 \cdot 10^{-23}$} & $6.35 \cdot 10^{-01}$ & $1.79 \cdot 10^{-01}$ \\
    0.80 & $6.68 \cdot 10^{-42}$ & $1.56 \cdot 10^{-82}$
    & N/A & \bm{$2.52 \cdot 10^{-33}$} & $5.25 \cdot 10^{-01}$ & $2.61 \cdot 10^{-01}$ \\
    0.85 & $2.88 \cdot 10^{-49}$ & $9.78 \cdot 10^{-110}$
    & N/A & \bm{$4.77 \cdot 10^{-34}$} & $8.10 \cdot 10^{-02}$ & $1.57 \cdot 10^{-01}$ \\
    0.90 & $4.31 \cdot 10^{-62}$ & $2.24 \cdot 10^{-171}$
    & N/A & \bm{$5.29 \cdot 10^{-45}$} & $4.16 \cdot 10^{-01}$ & $8.74 \cdot 10^{-01}$ \\
    0.95 & $2.01 \cdot 10^{-68}$ & $2.00 \cdot 10^{-211}$
    & N/A & \bm{$2.31 \cdot 10^{-53}$} & $1.29 \cdot 10^{-01}$ & $9.41 \cdot 10^{-02}$ \\
    \end{tabular}
\end{table}

\begin{table}[h!]
    \caption{$\epsilon = 0.5$, $c = 0.2$}
    \footnotesize
    \begin{tabular}{c|cc|cccc}
    \bm{$p$}
    & \textbf{\begin{tabular}{@{}c@{}}1/e-value\\(non-private)\end{tabular}}
    & \textbf{\begin{tabular}{@{}c@{}}Exact p-value\\(non-private)\end{tabular}}
    & \textbf{\begin{tabular}{@{}c@{}}Ours, Laplace\\\end{tabular}}
    & \textbf{\begin{tabular}{@{}c@{}}Ours, Gauss.\\(\bm{$\delta=0.01$})\end{tabular}}
    & \textbf{\begin{tabular}{@{}c@{}}ToT atop\\1/e-value\end{tabular}}
    & \textbf{\begin{tabular}{@{}c@{}}ToT atop\\exact p-value\end{tabular}} \\
    \hline
    0.05 & $5.97 \cdot 10^{-69}$ & $3.48 \cdot 10^{-215}$
    & \bm{$4.87 \cdot 10^{-69}$} & $9.60 \cdot 10^{-68}$ & $5.57 \cdot 10^{-03}$ & $1.15 \cdot 10^{-02}$ \\
    0.10 & $4.88 \cdot 10^{-61}$ & $1.89 \cdot 10^{-165}$
    & $3.14 \cdot 10^{-60}$ & \bm{$3.24 \cdot 10^{-61}$} & $3.04 \cdot 10^{-03}$ & $1.33 \cdot 10^{-02}$ \\
    0.15 & $1.64 \cdot 10^{-47}$ & $8.64 \cdot 10^{-103}$
    & \bm{$8.51 \cdot 10^{-48}$} & $3.78 \cdot 10^{-47}$ & $1.47 \cdot 10^{-02}$ & $5.43 \cdot 10^{-02}$ \\
    0.20 & $2.35 \cdot 10^{-44}$ & $4.85 \cdot 10^{-91}$
    & \bm{$3.46 \cdot 10^{-44}$} & $2.13 \cdot 10^{-43}$ & $6.67 \cdot 10^{-03}$ & $3.11 \cdot 10^{-03}$ \\
    0.25 & $3.11 \cdot 10^{-31}$ & $3.03 \cdot 10^{-51}$
    & $8.49 \cdot 10^{-30}$ & \bm{$4.23 \cdot 10^{-31}$} & $5.30 \cdot 10^{-03}$ & $3.91 \cdot 10^{-03}$ \\
    0.30 & $1.45 \cdot 10^{-17}$ & $6.78 \cdot 10^{-23}$
    & $1.05 \cdot 10^{-17}$ & \bm{$5.91 \cdot 10^{-18}$} & $5.59 \cdot 10^{-03}$ & $2.11 \cdot 10^{-02}$ \\
    0.35 & $1.28 \cdot 10^{-19}$ & $2.14 \cdot 10^{-26}$
    & $1.66 \cdot 10^{-19}$ & \bm{$1.24 \cdot 10^{-19}$} & $3.05 \cdot 10^{-02}$ & $1.11 \cdot 10^{-03}$ \\
    0.40 & $2.36 \cdot 10^{-11}$ & $8.29 \cdot 10^{-14}$
    & $3.01 \cdot 10^{-11}$ & \bm{$2.75 \cdot 10^{-11}$} & $9.12 \cdot 10^{-01}$ & $8.38 \cdot 10^{-02}$ \\
    0.45 & $1.20 \cdot 10^{-02}$ & $5.61 \cdot 10^{-04}$
    & \bm{$1.69 \cdot 10^{-02}$} & $1.99 \cdot 10^{-01}$ & $6.30 \cdot 10^{-02}$ & $5.19 \cdot 10^{-01}$ \\
    0.50 & $1.00 \cdot 10^{+00}$ & $6.35 \cdot 10^{-01}$
    & $1.00 \cdot 10^{+00}$ & $1.00 \cdot 10^{+00}$ & $7.69 \cdot 10^{-01}$ & \bm{$7.52 \cdot 10^{-01}$} \\
    0.55 & $4.18 \cdot 10^{-02}$ & $2.14 \cdot 10^{-03}$
    & $6.37 \cdot 10^{-02}$ & \bm{$1.21 \cdot 10^{-02}$} & $3.91 \cdot 10^{-01}$ & $9.35 \cdot 10^{-01}$ \\
    0.60 & $2.03 \cdot 10^{-09}$ & $2.13 \cdot 10^{-11}$
    & \bm{$1.49 \cdot 10^{-09}$} & $6.53 \cdot 10^{-09}$ & $9.07 \cdot 10^{-02}$ & $3.05 \cdot 10^{-01}$ \\
    0.65 & $7.99 \cdot 10^{-21}$ & $1.45 \cdot 10^{-28}$
    & $9.04 \cdot 10^{-21}$ & \bm{$4.62 \cdot 10^{-21}$} & $1.41 \cdot 10^{-01}$ & $7.15 \cdot 10^{-02}$ \\
    0.70 & $3.14 \cdot 10^{-27}$ & $1.15 \cdot 10^{-41}$
    & \bm{$5.02 \cdot 10^{-27}$} & $1.10 \cdot 10^{-25}$ & $8.66 \cdot 10^{-03}$ & $1.28 \cdot 10^{-02}$ \\
    0.75 & $4.00 \cdot 10^{-36}$ & $1.70 \cdot 10^{-64}$
    & \bm{$3.28 \cdot 10^{-36}$} & $1.53 \cdot 10^{-35}$ & $2.27 \cdot 10^{-03}$ & $8.96 \cdot 10^{-03}$ \\
    0.80 & $3.36 \cdot 10^{-41}$ & $3.37 \cdot 10^{-80}$
    & \bm{$1.27 \cdot 10^{-41}$} & $1.23 \cdot 10^{-40}$ & $4.51 \cdot 10^{-02}$ & $3.13 \cdot 10^{-02}$ \\
    0.85 & $3.37 \cdot 10^{-51}$ & $9.60 \cdot 10^{-118}$
    & $5.67 \cdot 10^{-51}$ & \bm{$5.27 \cdot 10^{-51}$} & $1.14 \cdot 10^{-02}$ & $2.01 \cdot 10^{-02}$ \\
    0.90 & $1.60 \cdot 10^{-57}$ & $6.11 \cdot 10^{-147}$
    & $1.95 \cdot 10^{-57}$ & \bm{$3.18 \cdot 10^{-58}$} & $2.69 \cdot 10^{-03}$ & $7.39 \cdot 10^{-03}$ \\
    0.95 & $3.98 \cdot 10^{-69}$ & $1.86 \cdot 10^{-216}$
    & \bm{$4.23 \cdot 10^{-69}$} & $2.99 \cdot 10^{-68}$ & $4.41 \cdot 10^{-03}$ & $3.18 \cdot 10^{-03}$ \\
    \end{tabular}
\end{table}

\begin{table}[h!]
    \caption{$\epsilon = 0.1$, $c = 0.05$}
    \footnotesize
    \begin{tabular}{c|cc|cccc}
    \bm{$p$}
    & \textbf{\begin{tabular}{@{}c@{}}1/e-value\\(non-private)\end{tabular}}
    & \textbf{\begin{tabular}{@{}c@{}}Exact p-value\\(non-private)\end{tabular}}
    & \textbf{\begin{tabular}{@{}c@{}}Ours, Laplace\\\end{tabular}}
    & \textbf{\begin{tabular}{@{}c@{}}Ours, Gauss.\\(\bm{$\delta=0.01$})\end{tabular}}
    & \textbf{\begin{tabular}{@{}c@{}}ToT atop\\1/e-value\end{tabular}}
    & \textbf{\begin{tabular}{@{}c@{}}ToT atop\\exact p-value\end{tabular}} \\
    \hline
    0.05 & $2.08 \cdot 10^{-18}$ & $1.08 \cdot 10^{-229}$
    & \bm{$3.26 \cdot 10^{-18}$} & $4.24 \cdot 10^{-17}$ & $9.00 \cdot 10^{-01}$ & $2.58 \cdot 10^{-01}$ \\
    0.10 & $2.74 \cdot 10^{-15}$ & $1.33 \cdot 10^{-149}$
    & \bm{$2.63 \cdot 10^{-15}$} & $3.67 \cdot 10^{-14}$ & $7.91 \cdot 10^{-01}$ & $3.76 \cdot 10^{-02}$ \\
    0.15 & $6.63 \cdot 10^{-13}$ & $1.57 \cdot 10^{-105}$
    & \bm{$2.52 \cdot 10^{-13}$} & $5.72 \cdot 10^{-12}$ & $2.38 \cdot 10^{-01}$ & $1.41 \cdot 10^{-02}$ \\
    0.20 & $5.68 \cdot 10^{-12}$ & $4.85 \cdot 10^{-91}$
    & $6.95 \cdot 10^{-12}$ & \bm{$2.26 \cdot 10^{-13}$} & $6.64 \cdot 10^{-04}$ & $3.65 \cdot 10^{-02}$ \\
    0.25 & $2.78 \cdot 10^{-10}$ & $4.26 \cdot 10^{-68}$
    & \bm{$4.06 \cdot 10^{-10}$} & $1.84 \cdot 10^{-09}$ & $8.57 \cdot 10^{-01}$ & $9.85 \cdot 10^{-02}$ \\
    0.30 & $8.87 \cdot 10^{-07}$ & $3.58 \cdot 10^{-32}$
    & $1.71 \cdot 10^{-06}$ & \bm{$8.86 \cdot 10^{-08}$} & $1.82 \cdot 10^{-02}$ & $1.81 \cdot 10^{-02}$ \\
    0.35 & $7.30 \cdot 10^{-05}$ & $2.10 \cdot 10^{-18}$
    & \bm{$5.06 \cdot 10^{-05}$} & $1.03 \cdot 10^{-04}$ & $9.86 \cdot 10^{-01}$ & $5.53 \cdot 10^{-02}$ \\
    0.40 & $3.91 \cdot 10^{-03}$ & $3.02 \cdot 10^{-09}$
    & $6.99 \cdot 10^{-03}$ & \bm{$1.46 \cdot 10^{-03}$} & $7.44 \cdot 10^{-01}$ & $2.58 \cdot 10^{-01}$ \\
    0.45 & $2.41 \cdot 10^{-01}$ & $7.16 \cdot 10^{-03}$
    & $7.69 \cdot 10^{-01}$ & $1.00 \cdot 10^{+00}$ & \bm{$2.98 \cdot 10^{-01}$} & $5.13 \cdot 10^{-01}$ \\
    0.50 & $1.00 \cdot 10^{+00}$ & $9.75 \cdot 10^{-01}$
    & $1.00 \cdot 10^{+00}$ & $1.00 \cdot 10^{+00}$ & \bm{$5.67 \cdot 10^{-01}$} & $6.03 \cdot 10^{-01}$ \\
    0.55 & $3.38 \cdot 10^{-01}$ & $1.77 \cdot 10^{-02}$
    & $5.55 \cdot 10^{-01}$ & $5.71 \cdot 10^{-01}$ & \bm{$4.23 \cdot 10^{-02}$} & $4.66 \cdot 10^{-02}$ \\
    0.60 & $3.91 \cdot 10^{-03}$ & $3.02 \cdot 10^{-09}$
    & $5.88 \cdot 10^{-03}$ & \bm{$5.58 \cdot 10^{-04}$} & $7.84 \cdot 10^{-01}$ & $4.34 \cdot 10^{-01}$ \\
    0.65 & $2.63 \cdot 10^{-05}$ & $3.01 \cdot 10^{-21}$
    & $4.49 \cdot 10^{-05}$ & \bm{$1.66 \cdot 10^{-05}$} & $4.90 \cdot 10^{-01}$ & $2.57 \cdot 10^{-01}$ \\
    0.70 & $9.91 \cdot 10^{-08}$ & $1.69 \cdot 10^{-40}$
    & $3.86 \cdot 10^{-07}$ & \bm{$3.39 \cdot 10^{-07}$} & $4.10 \cdot 10^{-02}$ & $1.57 \cdot 10^{-01}$ \\
    0.75 & $3.71 \cdot 10^{-10}$ & $1.54 \cdot 10^{-66}$
    & \bm{$1.45 \cdot 10^{-09}$} & $1.63 \cdot 10^{-08}$ & $7.98 \cdot 10^{-01}$ & $9.63 \cdot 10^{-02}$ \\
    0.80 & $1.83 \cdot 10^{-11}$ & $1.02 \cdot 10^{-83}$
    & \bm{$9.61 \cdot 10^{-12}$} & $9.10 \cdot 10^{-11}$ & $7.43 \cdot 10^{-01}$ & $5.66 \cdot 10^{-02}$ \\
    0.85 & $7.71 \cdot 10^{-14}$ & $1.62 \cdot 10^{-121}$
    & \bm{$1.19 \cdot 10^{-13}$} & $5.74 \cdot 10^{-13}$ & $8.99 \cdot 10^{-01}$ & $9.04 \cdot 10^{-02}$ \\
    0.90 & $1.13 \cdot 10^{-15}$ & $7.96 \cdot 10^{-158}$
    & \bm{$1.66 \cdot 10^{-15}$} & $4.25 \cdot 10^{-14}$ & $9.31 \cdot 10^{-01}$ & $3.42 \cdot 10^{-01}$ \\
    0.95 & $3.75 \cdot 10^{-18}$ & $1.25 \cdot 10^{-221}$
    & \bm{$2.82 \cdot 10^{-18}$} & $4.03 \cdot 10^{-18}$ & $7.65 \cdot 10^{-01}$ & $1.62 \cdot 10^{-01}$ \\
    \end{tabular}
\end{table}

\begin{table}[h!]
    \caption{$\epsilon = 0.05$, $c = 0.05$}
    \footnotesize
    \begin{tabular}{c|cc|cccc}
    \bm{$p$}
    & \textbf{\begin{tabular}{@{}c@{}}1/e-value\\(non-private)\end{tabular}}
    & \textbf{\begin{tabular}{@{}c@{}}Exact p-value\\(non-private)\end{tabular}}
    & \textbf{\begin{tabular}{@{}c@{}}Ours, Laplace\\\end{tabular}}
    & \textbf{\begin{tabular}{@{}c@{}}Ours, Gauss.\\(\bm{$\delta=0.01$})\end{tabular}}
    & \textbf{\begin{tabular}{@{}c@{}}ToT atop\\1/e-value\end{tabular}}
    & \textbf{\begin{tabular}{@{}c@{}}ToT atop\\exact p-value\end{tabular}} \\
    \hline
    0.05 & $5.04 \cdot 10^{-18}$ & $9.79 \cdot 10^{-218}$
    & \bm{$5.71 \cdot 10^{-18}$} & $5.84 \cdot 10^{-18}$ & $3.57 \cdot 10^{-01}$ & $3.56 \cdot 10^{-02}$ \\
    0.10 & $9.31 \cdot 10^{-16}$ & $1.06 \cdot 10^{-159}$
    & \bm{$1.02 \cdot 10^{-15}$} & $1.61 \cdot 10^{-15}$ & $1.43 \cdot 10^{-01}$ & $5.31 \cdot 10^{-03}$ \\
    0.15 & $2.75 \cdot 10^{-13}$ & $7.00 \cdot 10^{-112}$
    & \bm{$2.94 \cdot 10^{-13}$} & $3.06 \cdot 10^{-13}$ & $8.85 \cdot 10^{-01}$ & $4.55 \cdot 10^{-03}$ \\
    0.20 & $6.90 \cdot 10^{-12}$ & $8.60 \cdot 10^{-90}$
    & \bm{$6.71 \cdot 10^{-12}$} & $6.86 \cdot 10^{-12}$ & $9.91 \cdot 10^{-01}$ & $5.43 \cdot 10^{-03}$ \\
    0.25 & $1.55 \cdot 10^{-10}$ & $2.81 \cdot 10^{-71}$
    & \bm{$1.49 \cdot 10^{-10}$} & $1.92 \cdot 10^{-10}$ & $4.83 \cdot 10^{-01}$ & $5.28 \cdot 10^{-03}$ \\
    0.30 & $5.87 \cdot 10^{-06}$ & $8.57 \cdot 10^{-26}$
    & $5.99 \cdot 10^{-06}$ & \bm{$3.75 \cdot 10^{-06}$} & $8.28 \cdot 10^{-01}$ & $1.52 \cdot 10^{-03}$ \\
    0.35 & $8.00 \cdot 10^{-05}$ & $3.71 \cdot 10^{-18}$
    & \bm{$7.42 \cdot 10^{-05}$} & $1.05 \cdot 10^{-04}$ & $5.95 \cdot 10^{-01}$ & $3.63 \cdot 10^{-02}$ \\
    0.40 & $3.00 \cdot 10^{-03}$ & $9.25 \cdot 10^{-10}$
    & \bm{$3.02 \cdot 10^{-03}$} & $5.10 \cdot 10^{-03}$ & $1.60 \cdot 10^{-01}$ & $8.82 \cdot 10^{-01}$ \\
    0.45 & $4.62 \cdot 10^{-01}$ & $3.98 \cdot 10^{-02}$
    & $6.66 \cdot 10^{-01}$ & \bm{$5.54 \cdot 10^{-01}$} & $9.64 \cdot 10^{-01}$ & $7.51 \cdot 10^{-01}$ \\
    0.50 & $1.00 \cdot 10^{+00}$ & $5.48 \cdot 10^{-01}$
    & $1.00 \cdot 10^{+00}$ & $8.77 \cdot 10^{-01}$ & $1.53 \cdot 10^{-01}$ & \bm{$6.41 \cdot 10^{-02}$} \\
    0.55 & $9.27 \cdot 10^{-02}$ & $4.42 \cdot 10^{-04}$
    & $1.10 \cdot 10^{-01}$ & \bm{$5.35 \cdot 10^{-02}$} & $7.74 \cdot 10^{-01}$ & $7.36 \cdot 10^{-01}$ \\
    0.60 & $5.54 \cdot 10^{-03}$ & $1.38 \cdot 10^{-08}$
    & $5.83 \cdot 10^{-03}$ & \bm{$3.82 \cdot 10^{-03}$} & $5.24 \cdot 10^{-01}$ & $1.55 \cdot 10^{-01}$ \\
    0.65 & $1.67 \cdot 10^{-04}$ & $3.03 \cdot 10^{-16}$
    & \bm{$1.71 \cdot 10^{-04}$} & $1.84 \cdot 10^{-04}$ & $5.56 \cdot 10^{-01}$ & $6.10 \cdot 10^{-02}$ \\
    0.70 & $7.78 \cdot 10^{-06}$ & $6.65 \cdot 10^{-25}$
    & $8.23 \cdot 10^{-06}$ & \bm{$6.01 \cdot 10^{-06}$} & $3.64 \cdot 10^{-01}$ & $6.52 \cdot 10^{-04}$ \\
    0.75 & $4.59 \cdot 10^{-09}$ & $6.20 \cdot 10^{-54}$
    & $4.55 \cdot 10^{-09}$ & \bm{$3.25 \cdot 10^{-09}$} & $1.06 \cdot 10^{-01}$ & $6.70 \cdot 10^{-03}$ \\
    0.80 & $1.83 \cdot 10^{-11}$ & $1.02 \cdot 10^{-83}$
    & $1.85 \cdot 10^{-11}$ & \bm{$1.65 \cdot 10^{-11}$} & $5.40 \cdot 10^{-01}$ & $2.58 \cdot 10^{-03}$ \\
    0.85 & $1.95 \cdot 10^{-14}$ & $1.51 \cdot 10^{-132}$
    & $2.00 \cdot 10^{-14}$ & \bm{$1.73 \cdot 10^{-14}$} & $5.20 \cdot 10^{-01}$ & $2.43 \cdot 10^{-03}$ \\
    0.90 & $9.31 \cdot 10^{-16}$ & $1.06 \cdot 10^{-159}$
    & \bm{$1.03 \cdot 10^{-15}$} & $1.03 \cdot 10^{-15}$ & $4.89 \cdot 10^{-01}$ & $2.93 \cdot 10^{-02}$ \\
    0.95 & $2.08 \cdot 10^{-18}$ & $1.08 \cdot 10^{-229}$
    & $2.55 \cdot 10^{-18}$ & \bm{$2.42 \cdot 10^{-18}$} & $5.72 \cdot 10^{-01}$ & $1.88 \cdot 10^{-03}$ \\
    \end{tabular}
\end{table}

\section{Experiment details}

\subsection{Section 4.1 in the main text}

We use the e-value for the mean with distribution $F \sim \mathrm{Unif}(-1, 1)$ independent of $\theta$.
We use the first $100\,000$ samples of our dataset.

For the plot and confidence interval algorithm, we compute e-values for the mean over $k=50$ values of $\theta$, evenly spaced within $(0, 1)$.
For the confidence intervals we use a significance level of $\alpha = 0.05$.

\subsection{Section 4.2 in the main text}

We start by splitting the first $100\,000$ samples of the dataset into three splits: training (45\%), validation (15\%) and test (40\%).
We train a random forest classifier on the training data, and compute a validation loss on the validation set.
With this we define the desired safety threshold to be the validation 0-1 loss plus a tolerance of $0.05$.

We use the e-value for the mean with distribution $F \sim \mathrm{Unif}(0, c/\theta)$, where $\theta$ is the determined safety threshold and $c = 0.2$.
For the anytime-validity, we use a batch size of $128$ samples per batch.
The null is rejected at a significance level of $\alpha = 0.05$, i.e., when the e-value crosses $1/0.05 = 20$.

\subsection{Section 4.3 in the main text}

We split the data into three splits: training (54\%), calibration (36\%) and testing (10\%).
On the training data we fit a XGBoost classifier.
Then, on the calibration data we run our procedure, with conformity score $(x, y) \mapsto 1 / \mathrm{clamp}_{[0.01, 1]}(\widehat{p}(y | x))^{1/4}$.
We then linearly transform this score to have support in $[1, 100]$, and quantize it into 500 uniform bins over $[1, 100]$.
Atop this, we run our method.
Finally, we assess the resulting predictive interval sizes in the test split.
Empty prediction sets (if any) are converted into singleton predictive sets.

\end{document}